\documentclass[10pt,journal,compsoc]{IEEEtran}
\IEEEoverridecommandlockouts
\pdfoutput=1
\usepackage[linesnumbered, vlined, ruled]{algorithm2e}
\usepackage{amsmath,amssymb,amsfonts}
\usepackage{array}
\usepackage[utf8]{inputenc} 
\usepackage{booktabs} 
\usepackage{cite}
\usepackage{graphicx}
\usepackage[noend]{algpseudocode}
\usepackage[table,xcdraw]{xcolor}
\usepackage{textcomp}
\usepackage{subfigure} 
\usepackage{amsthm}
\usepackage{soul}
\newtheorem{theorem}{Theorem}
\newtheorem{lemma}{Lemma}

\definecolor{BrickRed}{RGB}{178,34,34} 

\def\BibTeX{{\rm B\kern-.05em{\sc i\kern-.025em b}\kern-.08em
    T\kern-.1667em\lower.7ex\hbox{E}\kern-.125emX}}

\makeatletter
\def\ps@IEEEtitlepagestyle{%
  \def\@oddfoot{\mycopyrightnotice}%
  \def\@evenfoot{}%
}
\def\mycopyrightnotice{%
  {\footnotesize 
  \begin{minipage}{\textwidth}
  \centering
  \textbf{This work has been submitted to the IEEE for possible publication. Copyright may be transferred without notice, after which this version may no longer be accessible.}
  \end{minipage}
    \hfill}
  \gdef\mycopyrightnotice{}
}

\begin{document}

\title{wChain: A Fast Fault-Tolerant Blockchain Protocol for Multihop Wireless Networks}

\author {Minghui~Xu,~\IEEEmembership{Student Member,~IEEE,}
    Chunchi~Liu, 
    Yifei~Zou, 
    Feng~Zhao,~\IEEEmembership{Member,~IEEE,}
    Jiguo~Yu,~\IEEEmembership{Senior Member,~IEEE,}
    Xiuzhen~Cheng,~\IEEEmembership{Fellow,~IEEE,}
    
\thanks{M. Xu is with the Department of Computer Science, The George Washington University, Washington, DC 20052 USA. E-mail: mhxu@gwu.edu.}
\thanks{C. Liu, Y. Zou, X. Cheng are with the School of Computer Science and Technology, Shandong University, Qingdao, 266510, P.R. China. E-mail: \{liuchunchi@sdu.edu.cn; yfzou@sdu.edu.cn; xzcheng@sdu.edu.cn\}.}
\thanks{F. Zhao (Corresponding Author) is with the Guangxi Colleges and Universities Key Laboratory of Complex System Optimization and Big Data Processing, Yulin Normal University, Yulin, P.R. China. E-mail: zhaofeng@guet.edu.cn.}
\thanks{J. Yu is with the Qilu University of Technology (Shandong Academy of Sciences), Jinan, Shandong, 250353, P.R. China; with Shandong Computer Science Center (National Supercomputer Center in Jinan), Jinan, Shandong, 250014, P.R. China; and with Shandong Laboratory of Computer Networks, Jinan, 250014, China.
Email: jiguoyu@sina.com.}
}

\markboth{Journal of \LaTeX\ Class Files,~Vol.~14, No.~8, August~2015}%
{Shell \MakeLowercase{\textit{et al.}}: Bare Demo of IEEEtran.cls for IEEE Journals}

\maketitle

\begin{abstract}
This paper presents $\mathit{wChain}$, a blockchain protocol specifically designed for multihop wireless networks that deeply integrates wireless communication properties and blockchain technologies under the realistic SINR model. We adopt a hierarchical spanner as the communication backbone to address medium contention and achieve fast data aggregation within $O(\log N\log\Gamma)$ slots where $N$ is the network size and $\Gamma$ refers to the ratio of the maximum distance to the minimum distance between any two nodes. Besides, $\mathit{wChain}$ employs data aggregation and reaggregation, and node recovery mechanisms to ensure efficiency, fault tolerance, persistence, and liveness. The worst-case runtime of $\mathit{wChain}$ is upper bounded by $O(f\log N\log\Gamma)$, where $f=\lfloor \frac{N}{2} \rfloor$ is the upper bound of the number of faulty nodes. To validate our design, we conduct both theoretical analysis and simulation studies, and the results only demonstrate the nice properties of $\mathit{wChain}$, but also point to a vast new space for the exploration of blockchain protocols in wireless networks.
\end{abstract}


\section{Introduction}

\IEEEPARstart{I}{n} recent years, the popularity of 5G and IoT has arisen more security and privacy issues relevant to identity management, data sharing, and distributed computing in wireless networks. With the inception of Bitcoin, blockchain has been envisioned as a promising technology that can be utilized to support various applications such as online payments and supply chain due to its salient properties of decentralization, immutability, and traceability. 
Correspondingly, effort has been put on protecting wireless applications using the blockchain technology, e.g., mobile edge computing (MEC) \cite{feng2020joint}, intelligent 5G \cite{DBLP:journals/network/DaiXMCHZ19}, vehicular networking \cite{DBLP:journals/winet/MalikNHL20}, and wireless sensor networking (WSN) \cite{8936389}.The common idea of applying blockchain in wireless networks is to introduce trustlessness with blockchain so that functions such as identity management and data sharing become more efficient and secure. 

However, previous studies on blockchain-enabled wireless applications mostly focus on developing practical applications or proposing architectures based on existing blockchain protocols, which were originally designed for wired network applications and thus are not suitable for wireless scenarios. To defend this point of view, let's consider the state-of-the-art blockchain protocols. The concept of proof of physical resources has been widely adopted, e.g., Proof-of-Work (PoW), Proof-of-Space, Proof-of-Elapsed Time, Proof-of-Space Time. The noteworthy drawback of these protocols is that they require high electricity, storage, or specific hardware (e.g., Intel SGX), which wireless devices cannot provide. On the other hand, protocols based on virtual resources such as stake, reputation, or credibility, i.e., Proof-of-Stake, Delegated Proof-of-Stake, Proof-of-Authority, Proof-of-Reputation, etc., always have a complicated design in order to avoid the centralization of wealth or power, which justifies why there is still no such a protocol particularly designed for wireless networks. 

Another line of blockchain protocols, such as Ripple, Algorand, Tendermint, and Hotstuff, rely on message passing. They provide blockchain systems with safety and liveness in confronting faulty nodes or even Byzantine failures. However, when implemented in wireless networks, the following two major problems need to be addressed:
\begin{itemize}
\item Traditional protocols commonly used on the Internet are not efficient enough in wireless networks. For example, PBFT \cite{pbft} and Tendermint \cite{kwon2014tendermint} can reach a consensus within $O(N^2)$ successful transmissions where $N$ is the network size, while Hotstuff \cite{yin2019hotstuff} reduces complexity to $O(N)$ but increases additional overhead due to the introduction of the cryptographic tools. 
\item Transforming fault-tolerant or Byzantine fault-tolerant protocols used in wired networks to wireless environments remains a tricky problem. A recent work by Poirot \textit{et al.} presented a fault-tolerant wireless Paxos for low-power wireless networks \cite{poirot2019paxos}. However, their results still indicate that traditional consensus algorithms require many message exchanges and high bandwidths, which may not be available in wireless networks.  
\end{itemize}

In this paper, we consider designing a message-passing-based blockchain protocol for wireless networks. To achieve this goal, we primarily need to overcome the challenge of properly implementing the medium access control (MAC) layer for blockchain. A promissing idea is to adopt the abstract MAC layer service (absMAC) proposed by Kuhn \textit{et al.} \cite{kuhn2009abstract}, which provides low-level network functions to help design distributed algorithms in wireless networks. The most efficient implementation of absMAC was presented in \cite{yu2018exact}, which achieves the optimal bound for a successful transmission in $O(\Delta+\log N)$ slots via carrier sensing, where $\Delta$ is the maximum number of neighbors a node may have. Obviously, with this efficient implementation and a fault-tolerant consensus algorithm that requires an optimal $O(N)$ successful transmissions (e.g., Hotstuff), the best we can obtain is a protocol with a communication complexity as high as $O(N(\Delta+\log N))$ slots. 

%

Nevertheless, this is far less than satisfactory. One needs a blockchain protocol that can take into consideration the unique features of wireless networking by deeply integrating wireless communications with blockchain to achieve the necessary properties of efficiency, fault-tolerance, persistence, and liveness. Since the most basic primitives heavily used in wireless consensus algorithms are broadcasts and data aggregations, which are respectively responsible for disseminating and collecting opinions or votes from peers, we use a spanner structure to accelerate the consensus process. A spanner is a hierarchical communication backbone that can organize nodes carefully to speed up data aggregation and dissemination processes. It introduces a sparse topology in which only a small number of links need to be maintained such that efficiency and simplicity can be well-balanced. Due to the need for decentralization and practicality, our spanner is constructed in a distributed manner and works under the realistic Interference-plus-Noise-Ratio (SINR) model. Facilitated with the spanner structure, we develop the two primitives of data aggregation and reaggregation to handle interference with an adaptive power scheme directly. These two primitives are adopted by our fault-tolerant $\mathit{wChain}$ protocol to address faulty behaviors caused by fail-stop errors and dynamic topologies such that the properties of resource conservation, fault-tolerance, efficiency, persistence, and liveness can be achieved.

The main contributions of this paper are summarized as follows.
\begin{enumerate}
\item To the best of our knowledge, $\mathit{wChain}$ is the first blockchain protocol that is particularly designed for multihop wireless networks under a realistic SINR model to deeply integrate wireless communications and blockchain technologies. 
\item The $\mathit{wChain}$ protocol ensures high performance by employing a spanner as the communication backbone. The runtime upper bound of the protocol is $O(\log N\log\Gamma)$ when crash failures happen in a low frequency, and the worst-case upper bound is $O(f\log N\log\Gamma)$.
\item Our $\mathit{wChain}$ protocol simultaneously achieves properties of resource conservation, fault-tolerance, efficiency, persistence, and liveness, which are proved by theoretical analysis and verified by simulation studies. 
\end{enumerate}

The rest of the paper is organized as follows. Section~\ref{sec:related:work} introduces the most related work on the state-of-the-art blockchain protocols and consensus algorithms in wireless networks. Section~\ref{sec:model} presents our model and preliminary knowledge. In Section~\ref{sec:protocol}, building blocks including utilities and data aggregation and reaggregation subroutines are first presented, then the three-phase $\mathit{wChain}$ protocol is explained in detail. Our $\mathit{wChain}$ protocol is theoretically analyzed in Section~\ref{sec:analysis} in terms of efficiency, persistence, and liveness. We report the results of our simulation studies in Section~\ref{sec:simulation} and conclude this paper in Section~\ref{sec:discussion}.

\section{Related Work}
\label{sec:related:work}
\textbf{Blockchain protocols for wireless networks.} Blockchain technology has been studied for wireless applications such as mobile edge computing (MEC) \cite{feng2020joint}, intelligent 5G \cite{DBLP:journals/network/DaiXMCHZ19}, vehicular networking \cite{DBLP:journals/winet/MalikNHL20}, secure localization \cite{8936389}, and Wireless D2D Transcoding \cite{liu2020deep}. 
Feng \textit{et al.} \cite{feng2020joint} considered the joint optimization of blockchain and MEC through a radio and computational resource allocation framework. Dai \textit{et al.} \cite{DBLP:journals/network/DaiXMCHZ19} proposed a secure and intelligent architecture for next-generation wireless networks by integrating blockchain and AI technologies. In vehicular ad hoc networks, Malik \textit{et al.} \cite{DBLP:journals/winet/MalikNHL20} utilized blockchain for secure key management. A blockchain-based trust management model was proposed to ensure secure localization in wireless sensor networks \cite{8936389}. A blockchain-enabled Device-to-Device (D2D) transcoding system was developed to provide trustworthy wireless transcoding services in \cite{liu2020deep}. Onireti \textit{et al.} \cite{onireti2019viable} provided an analytic modeling framework to obtain a viable area for wireless PBFT-based blockchain networks. In \cite{ren2018incentive}, a trustless mechanism with PoW-based blockchain was established to incentivize nodes to store data. Liu \textit{et al.} \cite{liu2018joint} realized computation offloading and content caching with mobile edge nodes in wireless blockchain networks. Sun \textit{et al.} \cite{sun2019blockchain} proposed an analytic framework to explore how the performance and security of wireless blockchain systems are affected by wireless communication features  such as SINR. 

Despite these extensive studies on applying blockchain to wireless networks, we are still in dire need of blockchain protocols that are specifically designed for wireless networks, which drives us to investigate and summarize existing distributed leader election and consensus algorithms for wireless networks in the sequel. 

\textbf{Consensus and leader election algorithms for wireless networks.} Consensus and leader election have also been extensively explored in various wireless contexts, and the corresponding solutions can guide us to design wireless blockchains in a proper way. 
Most existing studies on consensus and leader election for wireless networks assume a high-level wireless network abstraction \cite{moniz2012Byzantine} \cite{newport2014consensus} \cite{newport2018fault} \cite{dong2009resilient} \cite{vasudevan2003leader} \cite{raychoudhury2008top} or a realistic model grappling with issues in physical and link layers \cite{chockler2005consensus} \cite{scutari2008distributed} \cite{aysal2009reaching} \cite{richa2011self} \cite{golebiewski2009towards}. 

Moniz \textit{et al.} \cite{moniz2012Byzantine} proposed a BFT consensus protocol with runtime bounded by $O(N^2)$ among $k>\lfloor \frac{N}{2} \rfloor$ nodes in wireless ad hoc networks. They hid physical layer information but let nodes directly use high-level communication primitives. Leveraging the elegance of absMAC, Newport provided upper and lower bounds for distributed consensus in wireless networks \cite{newport2014consensus}. Subsequently, Newport and Robinson proposed a fault-tolerant consensus algorithm that terminates within $O(N^3\log N)$ with unknown network size \cite{newport2018fault}. A fully distributed leader election scheme was proposed to make election value unforgeable and be resistant to jamming attacks in wireless sensor networks. A secure extrema finding algorithm was proposed as a lock-step leader election algorithm that can be transformed to a secure preference-based leader election algorithm with a utility function depicting nodes' preference in wireless ad hoc networks \cite{vasudevan2003leader}. A top $k$-leader election algorithm with faulty nodes was presented in \cite{raychoudhury2008top}.

Chockler \textit{et al.} \cite{chockler2005consensus} investigated the relationship between collision detection and fault-tolerant consensus under a graph-based model. Assuming a graph model with message delays, Scutari and Sergio \cite{scutari2008distributed} proposed a consensus algorithm in wireless sensor networks with multipath fading. Aysal \textit{et al.} \cite{aysal2009reaching} studied the average consensus problem with probabilistic broadcasts under a graph-based model. Richa \textit{et al.} \cite{richa2011self} focused on self-stabilization of leader election for single-hop wireless networks to mitigate jamming attacks by adaptively adjusting the transmission probability at the MAC layer. Go{\l}{\c e}biewski and Klonowski \cite{golebiewski2009towards} proposed a fair leader election scheme in ad-hoc single-hop radio sensor networks, enabling resistance to adversaries who can transmit continuously to block channels or try to forge identities. 

In contrast, we observe that designing a blockchain protocol is more challenging than consensus and leader election in wireless networks. On the one hand, blockchain introduces new data structures such as a chain of blocks, and thus the corresponding computation, communication, and storage overheads should be carefully addressed. On the other hand, it is harder to balance performance and security in blockchain since consensus and leader elections offer more straightforward services without requiring a strong security guarantee. More importantly, Blockchain requires extra design elements (e.g., randomness, transaction and block verification, blockchain update) to guarantee strict persistence and liveness properties. 

\section{Models and Preliminaries}
\label{sec:model}

\textbf{Blockchain Basics.} Each node $v$ maintains a blockchain locally, denoted by $BC_v$, which is a hash-chain of blocks. $B_v^i$ refers to the $i$th block in $BC_v$. We also denote $BC_v^{i+}$ ($BC_v^{i-}$) as the partial blockchain of $BC_v$ before (after) $B_v^i$. Each block contains multiple transactions, and $tx_i^j$ stands for the $j$th transaction in $B_v^i$. Assume the latest block of $BC_v$ is $B_v^k$; then  $v$' $view$ is defined as a tuple $\{seq, hash\}_v$, where $seq$ and $hash$ are the sequence number and block hash of $B_v^k$, respectively. Besides, we adopt the UTXO model due to its remarkable properties such as large scalability and high level of security. We further assume nodes are supported by public key infrastructure, and the cryptographic primitives such as signature leveraged in our design are secure so that no malicious entity can spoof the messages. 

\textbf{Network Model.} We consider a multihop wireless ad hoc network with a set $V$ of $N$ nodes deployed in a 2-dimensional geographic plane. Let $d(u,v)$ denote the Euclidean distance between nodes $u$ and $v$, $D_R(v)$ denote the disk centered at $v$ with a radius $R$, and $N_R(v)$ denote the set of nodes excluding $v$ within $D_R(v)$. For simplicity, we normalize the minimum distance between any two nodes to be 1, and denote by $\Gamma$ the ratio of the maximum distance to the minimum distance between any two nodes. Each node has a unique id and knows no advanced information other than the network size $N$. The transmission power of each node can be controlled for interference mitigation. 
Assume nodes can crash at any time, which means that each node is either functioning normally or completely stop working. A node is regarded as \textit{faulty} if it crashes in the current epoch or does not have the latest view due to crash failures in previous epochs.  Without loss of generality, we assume $N$ is odd. Our protocol can tolerate at most $f$ faulty nodes where $f=\lfloor \frac{N}{2} \rfloor$. 


\textbf{Interference and SINR Model.} We adopt the Signal-to-Interference-plus-Noise-Ratio (SINR) model, which captures the wireless network interference in a more realistic and precise way than a graph-based one. A standard SINR model can be formulated as follows, which states that a message sent by $u$ is correctly received by $v$ if and only if 
\begin{equation}
\begin{aligned}
SINR(u,v)=\frac{P_u\cdot d(u,v)^{-\alpha}}{\mathcal{N}+\sum_{w\in S\setminus\{u\}}P\cdot d(w,v)^{-\alpha}}\geq \beta
\end{aligned}
\end{equation}
holds, where $\mathcal{N}$ is the ambient noise, $\alpha\in (2,6]$ is the path-loss exponent, threshold $\beta>1$ is determined by hardware, and $S\subseteq V$ denotes the set of nodes transmitting simultaneously with $u$. Besides, we further assume that nodes can perform physical carrier sensing. 

\textbf{Maximal Independent Set (MIS).} 
A set $S\subseteq V$ is an \textit{independent set} of $V$ with respect to distance $r$ if for any pair of nodes $u$ and $v$ in $S$, $d(u,v)>r$; and $S$ is referred to as a \textit{maximal independent set} if for any node $w\notin S$, there is a node $x\in S$ such that $d(w,x)\leq r$. MIS has been widely researched in recent years, and there exist a number of  methods computing an MIS in a distributed manner. In this paper, we adopt the approach presented in \cite{7889052}, which computes a distributed MIS in optimal time $O(\log N)$ if the nodes' density is a constant. 

\begin{figure*}[!htb]
\centering
\subfigure[Construction of the first two layers of a spanner followed by an MIS; here hollow and solid circles represent $V_0$ and $V_1$, respectively.]{
    \label{fig:spanner_two_layer}
    \centering
    \includegraphics[width=0.35\textwidth]{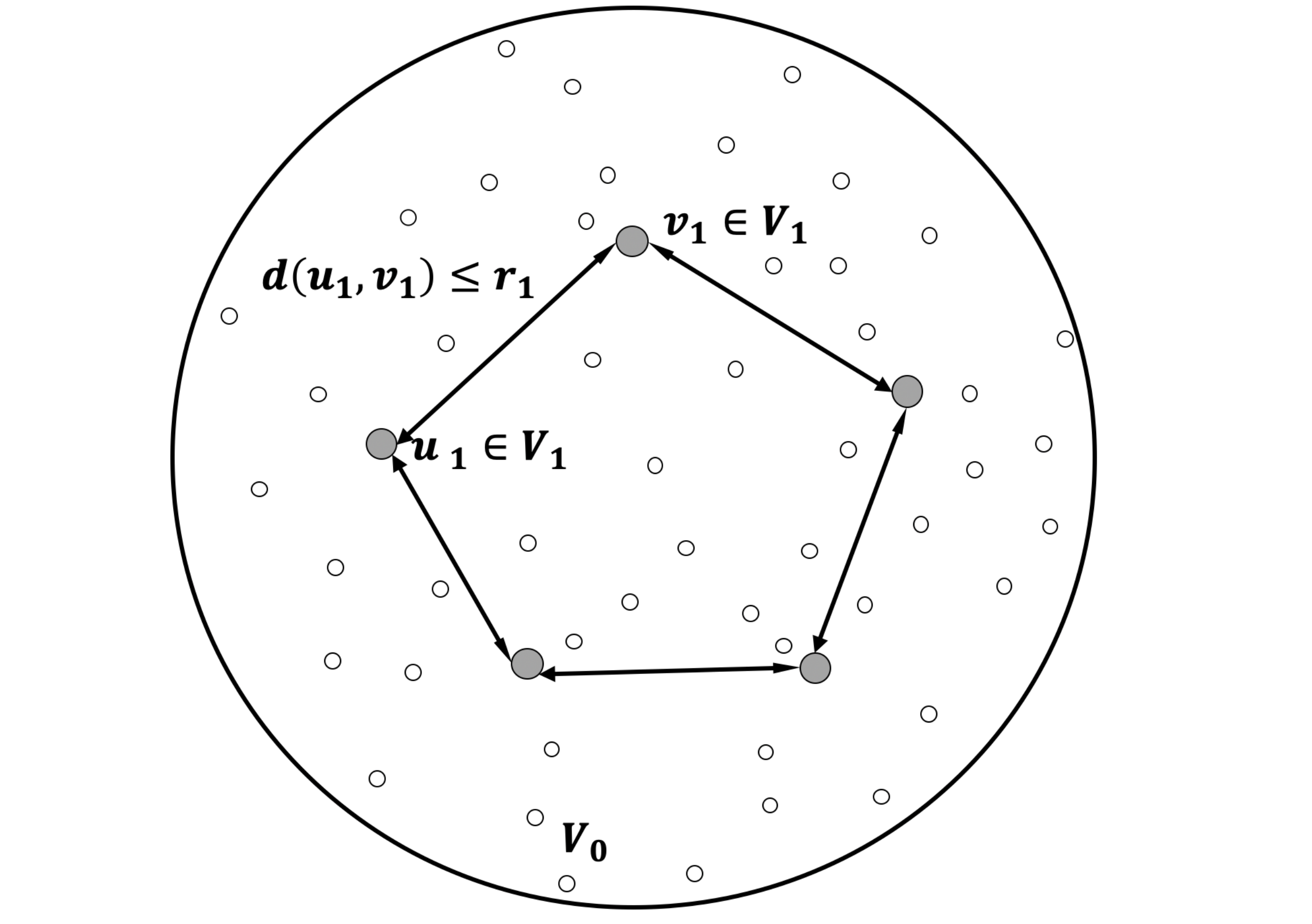}
}%
\hspace{.8in}
\subfigure[A ($\log\Gamma+1$)-level full spanner where arrows stand for the MIS relationship: each MIS node at a certain level covers a set of nodes in the next lower level.]{
    \label{fig:spanner__layer}
    \centering
    \includegraphics[width=0.35\textwidth]{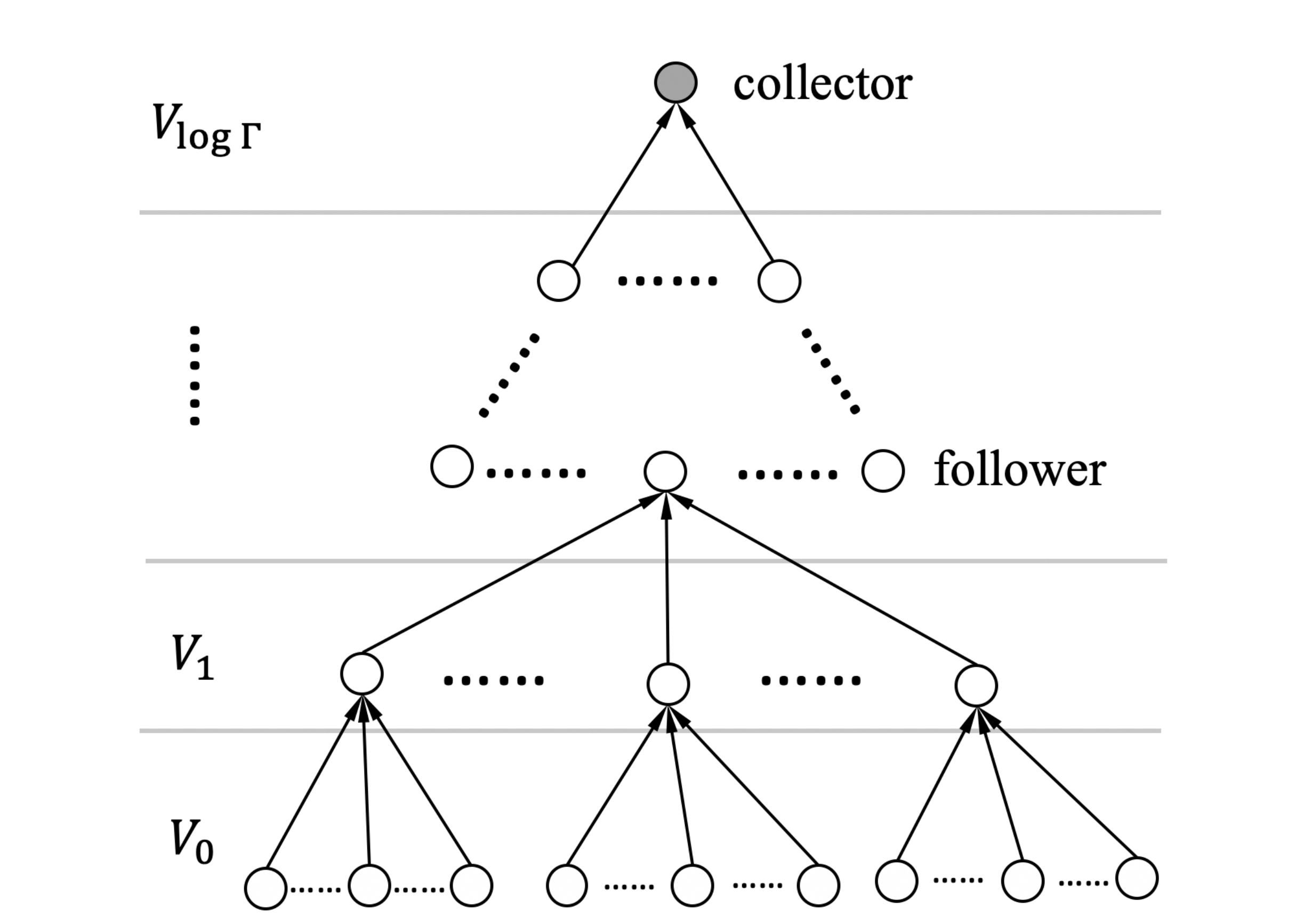}
}%
\caption{Spanner visualization.}
\label{fig:spanner}
\end{figure*}

\textbf{Spanner Construction.} A spanner is a network backbone possessing the following properties: it only needs to maintain a small number of links, and can balance well between efficiency and simplicity compared to other topologies. Taking advantages of these features we employ a spanner to facilitate the deployment of our data aggregation algorithm in $\mathit{wChain}$. More specifically, we adopt the distributed spanner construction algorithm presented in \cite{7889052}, which can construct a sparse spanner, denoted as $H$, with a bounded maximum degree, in $O(\log N \log \Gamma)$ slots with a high probability. 
As illustrated in Fig.~\ref{fig:spanner}, the construction process of $H$ contains $\log\Gamma$ rounds; and in the $i$th round, where $i=1,2,\cdots,\log\Gamma$, $V_i$ contains the nodes in a maximal independent set elected from $V_{i-1}$ by running a distributed MIS algorithm with respect to $r_i=2^i$; thereby $V_{\log\Gamma}\subseteq\cdots\subseteq V_1 \subseteq V_0=V$. One can see that the constructed $H$ holds the following properties: 
\begin{itemize}
  \item nodes in $V_i$ constitute an MIS of $V_{i-1}$ with respect to $r_i$; 
  \item each node $v\in V_{i-1}\setminus V_i$ has a parent node $u\in V_i$ and $d(v,u)\leq r_i$;
  \item $V_{\log\Gamma}$ contains only one node, i.e., the root.
\end{itemize}

For ease of explanation, nodes in $V_0\setminus V_{\log\Gamma}$ and $V_{\log\Gamma}$ are referred to as \emph{followers} and \emph{collector}, respectively. In this paper, we say that an event $E$ occurs with high probability (w.h.p.) if for any $c\geq1$, $E$ occurs with probability at least $1-1/N^c$. A summary of all critical notations and their semantic meanings is presented in Table~\ref{table:notation}.

\begin{table}[!htb]
\caption{Summary of Notations} \label{table:notation}
\centering
\begin{tabular}{c|c}
  \hline
  \textbf{Symbol} & \textbf{Description} \\
  \hline
    $BC_v$ & the blockchain locally stored at node $v$\\
    $B^{i}_v$ & the $i$th block in $BC_v$ \\
    $BC_v^{i+}$ & the partial blockchain of $BC_v$ before $B_v^i$\\
    $BC_v^{i-}$ & the partial blockchain of $BC_v$ after $B_v^i$\\
    $C_v$ & the set of child nodes of $v$\\
    $d_v$ & $v$'s data (to be aggregated)\\
    $id_v$ & $v$'s unique id\\
    $id_v^p$ & the unique id of $v$'s parent\\
    $M_v$ & the message queue maintained by $v$\\
    $m_v$ & a single message sent by $v$\\
    $p$ & the uniform transmission probability\\
    $P_i$ & the transmission power in the $i$th round\\
    $\hat{P}$ & the maximum transmission power\\
    $R_i$ & the $i$th round \\
    $tx_i^j$ & the $j$th transaction in the $i$th block\\
    $V_i$ & the set of nodes as an MIS of $V_{i-1}$\\
    $\sigma$ & a sufficiently large constant to determine $p$\\
    $\lambda'$ & an upper bound of the network density\\
    $\mu$ & a constant to determine the round length\\
    \hline
\end{tabular}
\end{table}

\section{The Protocol}
\label{sec:protocol}
In this section, we first present an overview on $\mathit{wChain}$ and demonstrate the involved utilities; then we detail the data aggregation and reaggregation algorithms, the two subroutines of our blockchain protocol; finally, we propose the three-phase fast fault-tolerant $\mathit{wChain}$ protocol.

\subsection{Technical Overview and Utilities}
\label{sec:sub:overview}
\subsubsection{Technical Overview}
The $\mathit{wChain}$ protocol is executed in disjoint and consecutive time intervals called \textit{epochs}, and at each epoch, no more than one block can be generated. Non-faulty nodes should append the new block to their local blockchain so that they can jointly maintain a consistent global view. Within each epoch, a spanner is first established as a communication backbone, and the collector of the spanner is appointed as the leader to take charge of $\mathit{wChain}$ for the entire epoch. However, if a crash failure occurs, a new spanner should be constructed in a reaggregation procedure for the same epoch. Under this circumstance, the collector might be changed, but the leader holds the line. Data can be aggregated from followers to the new collector who subsequently sends the aggregated data to the leader to complete the data aggregation process of the current epoch. $\mathit{wChain}$ proceeds by three phases, namely \texttt{PREPARE}, \texttt{COMMIT}, and \texttt{DECIDE}. In the  \texttt{PREPARE} phase, an incumbent leader aggregates view messages from all followers to learn about the latest view. If more than $f$ followers respond with the same view as that of the leader, the leader can aggregate transactions in the \texttt{COMMIT} phase. Otherwise, the leader sends a message to the followers to abandon the current epoch. The threshold $(f+1)$ (including $f$ followers plus one leader) is specifically designated to satisfy the \textit{quorum intersection property} which implies that among the $(f+1)$ nodes, at least one node is non-faulty and has the correct information. In the  \texttt{DECIDE} phase, the leader verifies the collected transactions, organizes them into a new block, and then sends the new block together with the view update information to the entire network, by which the nodes can update their local blockchain to obtain the latest view. 

\subsubsection{Utilities}
\begin{algorithm}[!htbp]
\label{alg:utility}
\DontPrintSemicolon
\caption{Utilities for node $v$}
\SetKwProg{Fn}{Function}{}{}
\Fn{MSG($data_v$)}{
    $m.data\leftarrow data_v$\;
    $m.timestamp\leftarrow time$\;
    $m.kindred\leftarrow \{id_v, id_v^p\}$\;
    $m.role\leftarrow \{role_v, level_v\}$\;
    \Return $m$\;
}
\Fn{$add(M_v, M_u)$}{
    \Return $M_v\leftarrow M_v\cup M_u$
}
\Fn{$packup(M_v)$}{
    \For {$m\in M_v$}{
        \If {$m.data$ is a valid transaction}{
            organize $m.data$ into $B_v$\;
        }
    }
    \Return $B_v$\;
}
\Fn{$append(BC_v, B_u)$}{
    \If {$B_u$ contains the hash of the $B_v^k$}{
        \Return $BC_v\leftarrow BC_v+B_u$\;
    }
    
}
\Fn{$extract(BC_v, M_v^{view})$}{
    $i\leftarrow$ the $(f+s)$th highest $seq$ searched in $M_v^{view}$\;
    \Return $BC_v^{i+}$\;
}
\Fn{$update(BC_v, BC_u^{i+}$)}{
    \For {$j=1$ to $|BC_u^{i+}|$}{
        \If {$B_u^j\notin BC_v$}{
            $append(BC_v, B_u^j)$\;
        }
    }
    \Return $BC_v$
}
\end{algorithm}

Before delving into details of the $\mathit{wChain}$ protocol, we explain its commonly used utilities. First of all, \textit{MSG($data_v$)} is used the most often in $\mathit{wChain}$ to generate a single message $m$ embodying a variable data field, which might be a string, a transaction, or a block. For example, the parameter $data_v$ can be $view_v$ (view information), or $tx_v$ (transaction); it also indicates that the input data is signed by $v$. Other than the data field, $m$ also includes the \texttt{timestamp}, \texttt{kindred}, and \texttt{role} fields for the verification purpose. The \texttt{kindred} is a tuple consisting of $v$'s identity $id_v$ and the identity $id_v^p$ of $v$'s parent. The \texttt{role} field clarifies the role (i.e., follower, collector, or leader) and the specific level information such as $V_0\setminus V_1$. A receiver recognizes $m$ as valid if $m$ includes correct identity and role information. Moreover, $M_v$ is a message queue to hold multiple messages. Duplicate messages in $M_v$, which can be identified based on the timestamp information,  should be discarded. 
A node $v$ can perform an $add(M_v, M_u)$ operation to append $M_u$ to its local $M_v$. For block formation, we provide $packup(M_v)$ to read all messages from $M_v$ and organize valid transactions into a block. Besides, $append(BC_v, B_u)$ can append a block $B_u$ from node $u$ to $BC_v$ only if $B_u$ contains the hash of the last block of $BC_v$. In addition, $extract(BC_v, M_v^{view})$ intends to extract a partial blockchain $BC_v^{i+}$ from $BC_v$ where $i$ is determined by the $(f+s)$th highest $seq$ searched in $M_v^{view}$ and $s$ is an adjustable constant to be determined later. For node recovery, the $update(BC_v, BC_u^{i+})$ function is designed to help a node $v$ update $BC_v$ by complementing missed blocks from the received $BC_u^{i+}$.

\subsection{Data Aggregation and Reaggregation}
\label{sec:sub:aggregation}
In this subsection, we present two critical algorithms to realize data aggregation and reaggregation. The objective of data aggregation is to rapidly collect data from all followers to the leader in $O(\log n\log\Gamma)$ w.h.p. However, messages might be lost due to crash failures. Therefore, we propose the data reaggregation subroutine to remedy such a situation and ensure that the leader can completely aggregate the data from all non-faulty nodes within an epoch. 
%

\subsubsection{Broadcast-Oriented Communications}
\label{sec:sub:broadcast}
When blockchains are to be implemented in a wireless network, unicast and multicast are generally not needed. Instead, we should exploit the broadcast nature of the wireless medium. Decomposing a typical consensus process one can see that there exist three major communication patterns, namely one-to-many, many-to-one, and many-to-many, that heavily employ broadcast and data aggregation, the two communication primitives. Broadcast can improve blockchain's efficiency since a one-to-many communication only costs a one-time broadcast to disseminate a node's message to all its peers within the communication range. For many-to-many, the communication complexity is $O(N^2)$ for a network size of $N$ if using unicast and multicast. With broadcast, the complexity can be reduced to $O(N)$.

However, to reach a consensus, medium contention and packet collision need to be carefully addressed in wireless networks. For this purpose we restrict the transmission probability to be $p$ and transmission power to be $P_i$, which are formally described in Section~\ref{sec:sub:aggregation} (also see line 6 in Algorithm.~\ref{alg:data:agg}).

\subsubsection{Data Aggregation}
\label{sec:sub:data:aggr}

\begin{algorithm}[!htbp]
\label{alg:data:agg}
\DontPrintSemicolon
\caption{$DataAggregation(data_v)$ Subroutine}
\SetKwInOut{Input}{input}
\SetKwProg{Fn}{Function}{}{}
\Fn{DataAggregation($data_v$)}{
    Initially, $m_v\leftarrow \textit{MSG($data_v$)}$, $M_v=\{m_v\}$\;
    $\triangleright$ \textcolor{blue}{In $R_i(i=1,2,\cdots,\log\Gamma)$:}\;
    \eIf{$v\in V_{i-1}\setminus V_i$}{
            \For {$\mu\cdot \log N$ slots}{
                send $M_v$ with probability $p=\frac{1}{\sigma\lambda'}$ and power $P_i=2\mathcal{N}\beta r_i^{\alpha}$\;
            }
    }{
        \If {$v\in V_i$}{
            \For {$\mu\cdot \log N$ slots}{
                listen on the channel\;
                \If {receive a valid $M_u$}{
                    $M_v\leftarrow add(M_v, M_u)$\;
                }
            }
        } 
    }
}

\end{algorithm}

Leveraging a spanner, we propose the \textit{DataAggregation($data_v$)} subroutine to aggregate the data level by level. As shown in Algorithm~\ref{alg:data:agg}, a node $v$ executing \textit{DataAggregation($data_v$)} takes $data_v$ as input, generates a message $m_v$ by \textit{MSG(data)}, and initializes $M_v=\{m_v\}$. Then data aggregation proceeds by $\log\Gamma$ rounds.  We denote by $R_i$ the $i$th round. Recall that for $i=1,2,\cdots,\log\Gamma$, a node $v\in V_{i-1}\setminus V_i$ has a parent node $u\in V_i$ and $d(v,u)\leq r_i$. As a consequence, $R_i$ is responsible for aggregating the data from $V_{i-1}\setminus V_i$ to $V_i$, and $R_{\log\Gamma}$ is the final round when the collector receives the data from $V_{\log\Gamma-1}\setminus V_{\log\Gamma}$. 

In a specific round $R_i$, the nodes in $V_{i-1}\setminus V_i$ constantly send $M_v$ for $\mu\cdot \log N$ slots with probability $p=\frac{1}{\sigma\lambda'}$, where $\lambda'=25$, $\mu$ and $\sigma$ are sufficiently large constants whose lower bounds are given in Section~\ref{sec:sub:efficiency}. Note that $\mu\cdot \log N$ is the optimal number of slots to ensure that without crash failures, the data from all child nodes can be completely aggregated to their parent w.h.p. The transmission power is set to be $P_i=2\mathcal{N}\beta r_i^{\alpha}$ so that the transmission range of each node is $2^{1/\alpha}r_i$, where $r_i=2^i$. This transmission range is slightly larger than $r_i$ such that for any child $v$, its parent is within the unit disk centered at $v$ with a radius $r_i$. This power control strategy improves child nodes' ability to resist interference outside the unit disk, thus contributing to the success of the transmissions. Moreover, a node $v\in V_i$ as a parent in $R_i$ listens on the channel for $\mu\cdot \log N$ slots to receive messages from its children. If $v$ receives a $M_u$ from a child $u$, it appends $M_u$ to $M_v$. The entire data aggregation process can be finished in $\mu\log N\log\Gamma$ slots, and the data from all non-faulty nodes can be aggregated w.h.p., which is proved in Section~\ref{sec:sub:efficiency}.

\subsubsection{Reaggragation}
\label{sec:sub:reag}
\begin{algorithm}[!htbp]
\label{alg:reaggregation}
\DontPrintSemicolon
\caption{$Reaggregation(data_v)$ Subroutine}
\SetKwProg{Fn}{Function}{}{}
$\triangleright$ \textcolor{BrickRed}{as a leader}\; 
\While {$true$}{
    $\triangleright$ \textcolor{blue}{slot one}\; 
    broadcast $M_{\ell}^{data}$\;
    $\triangleright$ \textcolor{blue}{slot two}\;
    listen on the channel\;
    $\triangleright$ \textcolor{blue}{slot three}\;
    \eIf {$\text{sense noise} > \mathcal{N}$ in slot two}{
        broadcast $m_{\ell}\leftarrow \textit{MSG}(reaggregation_{\ell})$\;
    }{
    	broadcast $m_{\ell}\leftarrow\textit{MSG}(stop_{\ell})$ and break\;
    }
    $\triangleright$ \textcolor{blue}{data reaggregation}\;
    wait for aggregated data from a collector\;
}
$\triangleright$ \textcolor{BrickRed}{as a follower}\; 
\While {$true$}{
	$\triangleright$ \textcolor{blue}{slot one}\; 
    listen on the channel\;
    $\triangleright$ \textcolor{blue}{slot two}\;
    \If {receive $M_{\ell}^{data}$ in slot one and $data_v\notin M_{\ell}^{data}$}{
        broadcast $m_v\leftarrow\textit{MSG}(miss_v)$\;
    }
    $\triangleright$ \textcolor{blue}{slot three}\;
    listen on the channel\;
    $\triangleright$ \textcolor{blue}{data reaggregation}\;
    \eIf {receive reaggregation message in slot three}{
        run $SpannerConstruction$\;
        run $DataAggregation(data_v)$\;
    }{
    	break\;
    }
}
\end{algorithm}

The reaggregation subroutine has two stages: a three-slot integrity check stage (lines 3-11, lines 16-22) and a data reaggregation stage (lines 12-13, lines 23-28). We define a broadcast operation (used in lines 4/9/11/20) as transmitting a message with $\hat{P}=2\mathcal{N}\beta r_{\log\Gamma}^{\alpha}$ so that a node listening on the channel can either receive a message from the sender or sense noise exceeding $\mathcal{N}$. Concretely, the integrity check stage intends to examine whether the leader loses any message from non-faulty nodes using physical carrier sensing. In slot one, the leader $l$ broadcasts its current $M_{\ell}^{data}$ to the entire network, where $M_{\ell}^{data}$ is the message queue embodying the messages whose type is $data$ (e.g., $M_{\ell}^{view}$ is the message queue of the view messages). Upon receiving $M_{\ell}^{data}$, each node $v$ examines if its $data_v$ is included in $M_{\ell}^{data}$. If not, $v$ broadcasts $m_v\leftarrow\textit{MSG}(miss_v)$ so that in slot two the leader can get the notice saying that some messages are missed by sensing noise greater than $\mathcal{N}$, and broadcast $m_{\ell}\leftarrow\textit{MSG}(reaggregation_{\ell})$ in slot three to start the second stage. In the data reaggregation stage, all nodes except the leader run \textit{SpannerConstrcution} to reconstruct a spanner free from faulty nodes. The \textit{SpannerConstrcution} procedure is the same as the one we illustrate in Section~\ref{sec:model}. The new spanner does not include the leader and it elects a new collector who is responsible for sending the aggregated data to the leader. Afterwards, the nodes whose messages are missed in $M_{\ell}^{data}$ run $DataAggregation(data_v)$. Only when no messages from non-faulty nodes are missed can the leader broadcast $m_{\ell}\leftarrow\textit{MSG}(stop_{\ell})$ to end the reaggregation process. 

\subsection{Fast Fault-Tolerant Blockchain Protocol}
\label{sec:sub:agree}

\begin{algorithm}[!htbp]
\label{alg:agreement}
\DontPrintSemicolon
\caption{Fast Fault-Tolerant Blockchain Protocol}
$\triangleright$ \textcolor{blue}{ \texttt{PREPARE}}\;
$\triangleright$ \textcolor{BrickRed}{as a leader}\; 
broadcast $m_{\ell}\leftarrow\textit{MSG}(view_{\ell})$\;
listen on the channel for $\mu\log N \log\Gamma$ slots\;
execute $Reaggregation(view_{\ell})$\;
$\triangleright$ \textcolor{BrickRed}{as a follower}\; 
\eIf {receive $view_u$ from a leader}{
    run $DataAggregation(view_v)$\;
}{
    abandon the current epoch\;
}
execute $Reaggregation(view_v)$\;
$\triangleright$ \textcolor{blue}{\texttt{COMMIT}}\;
$\triangleright$ \textcolor{BrickRed}{as a leader}\; 
\If {$|\{m\in M_{\ell}^{view}|m.data = view_{\ell}\}|\geq f+1$}{
    broadcast $m_{\ell}\leftarrow\textit{MSG}(correct_{\ell})$\;
    listen on the channel for $\mu\log N \log\Gamma$ slots\;
}
execute $Reaggregation(tx_v)$\;
$\triangleright$ \textcolor{BrickRed}{as a follower}\; 
\eIf {receive $correct_{\ell}$ from a leader}{
    run $DataAggregation(tx_v)$\;
}{ 
    abandon the current epoch\;
}
execute $Reaggregation(tx_v)$\;
$\triangleright$ \textcolor{blue}{ \texttt{DECIDE}}\;
$\triangleright$ \textcolor{BrickRed}{as a leader}\; 
$B_{\ell}\leftarrow \textit{packup}(M_{\ell}^{tx})$, and $BC_{\ell}\leftarrow append(BC_{\ell}, B_{\ell})$\;
broadcast $BC_{\ell}^{i+}\leftarrow extract(BC_{\ell}, M_{\ell}^{view})$\;
$\triangleright$ \textcolor{BrickRed}{as a follower}\; 
\eIf {receive $BC_{\ell}^{i+}$ from the leader}{
    update($BC_v$, $BC_{\ell}^{i+}$)\;
}{
    abandon the current epoch\;
}
\end{algorithm}

$\mathit{wChain}$ is a three-phase protocol that can achieve consensus on the sequence of blocks and handle failures caused by wireless node crashes. Concretely, $\mathit{wChain}$ leverages broadcast communications, data aggregation and reaggregation, which are all specifically designed for wireless networks. At each epoch, $\mathit{wChain}$ proceeds by three phases, namely \texttt{PREPARE}, \texttt{COMMIT},  and \texttt{DECIDE}. In the following, we depict each phase to demonstrate how fast fault-tolerance can be achieved in $\mathit{wChain}$.

\textbf{Prepare.} The  \texttt{PREPARE} phase intends to help a leader obtain a global view. Recall that when a spanner is constructed for the first time, $\mathit{wChain}$ appoints the collector as the leader to take charge of the current epoch. As a leader, $\ell$ broadcasts $m_{\ell}\leftarrow\textit{MSG}(view_{\ell})$ to the entire network. Each follower $v$ runs $DataAggregation(view_v)$ upon receiving the view information from ${\ell}$ in the previous slot. Otherwise, the follower abandons the current epoch. Note that all nodes should execute the reaggregation subroutine to ensure that the data from all non-faulty nodes are completely aggregated. 

\textbf{Commit.} Denote by $M_{\ell}^{view}$ the message queue of $\ell$ embodying its view information. The requirement of $|\{m\in M_{\ell}^{view}|m.data = view_{\ell}\}|\geq f+1$ means that the leader should successfully receive no less than $f$ view messages that have identical views as itself. That is also to say, at least $f+1$ nodes ($f$ followers and one leader) have an identical view. If such a requirement is satisfied, the leader can broadcast a correct message and listen on the channel to receive transactions while the followers receiving a correct signal start transaction aggregation. The reaggregation subroutine is still executed to ensure the full aggregation of transactions.

\textbf{Decide.} When the leader ${\ell}$ receives all transactions from non-faulty nodes, it packs up the transactions to form a block $B_{\ell}$, and appends the $B_{\ell}$ to its local $BC_{\ell}$. Then the leader ${\ell}$ executes $extract(BC_{\ell}, M_{\ell}^{view})$ to formulate a $BC_{\ell}^{i+}$, which is used to help recover at least $s$ nodes that have crashed in previous epochs and need to update their blockchains to become non-faulty. Each non-faulty follower $v$ updates $BC_v$ by running $update(BC_v, BC_{\ell}^{i+})$.

\section{Protocol Analysis}
\label{sec:analysis}
In this section, we analyze the protocol in terms of efficiency of data aggregation and reaggregation, persistence, and liveness.

\subsection{Efficiency of Data Aggregation and Reaggregation}
\label{sec:sub:efficiency}
\begin{theorem}
\label{thm:time}
  The runtime of the data aggregation subroutine is upper bounded by $O(\log N \log\Gamma)$ slots w.h.p., and the runtime of the reaggregation subroutine is upper bounded by $O(f\log N \log\Gamma)$ slots w.h.p.
\end{theorem}
\begin{proof}

We first present Lemma~\ref{lemma:cluster:com}, which focuses on one slot in a given round $R_i$. 
\begin{lemma}
\label{lemma:cluster:com}
For a given slot in round $R_i$, if $v\in V_i$ is a parent node of some node in $V_{i-1}$, for any node $u\in V_{i-1} \cap N_{r_i}(v)$, if $u$ transmits, $v$ can receive the message with a constant probability.
\end{lemma}

\begin{proof}
For a given slot in round $R_i$, we denote the aggregated transmission probability of $V_{i-1} \cap N_{r_i}(v)$ by $P_i(v)=\sum_{w\in V_{i-1} \cap N_{r_i}(v)} p_w$. Let's first prove that $P_i(v)$ can be bounded by $\frac{1}{\sigma}$, where $\sigma$ is a sufficiently large constant. Due to the property of the maximum independent set, the disks of radius $r_i/2$ centered at any node $w\in V_{i-1} \cap N_{r_i}(v)$ are disjoint. We define density $\lambda$ as the number of nodes in $V_{i-1} \cap N_{r_i}(v)$, then the upper bound of density is
\begin{equation}
\begin{aligned}
    \lambda\leq \frac{\pi[r_1+r_0/2]^2}{\pi(r_0/2)^2}= 25,
\end{aligned}
\end{equation}
where $\lambda=25$ when $i=1$. Since for each $w\in V_{i-1} \cap N_{r_i}(v)$, $p_w=\frac{1}{\sigma\lambda'}$, we have $P_i(v)=\lambda\cdot \frac{1}{\sigma\lambda'}\leq \frac{1}{\sigma}$, where $\lambda'=25$. 

Then we partition the whole space outside $D_{r_i}(v)$ into rings $R_j$ for $j\geq 1$, where $R_j$ is the ring with a distance in the range $[jr_i, (j+1)r_i]$ from $v$. Denote by $S_j$ the set of nodes in $V_{i-1}$ that also fall into $R_j$. Considering the property of maximum independent set, one can see that the disks of radius $r_i/2$ centered at the nodes in $R_j$ are disjoint. Then we have
\begin{equation}
\begin{aligned}
    |S_j|&\leq \frac{\pi[(j+1)r_i+r_i/2]^2-\pi[jr_i-r_i/2]^2}{\pi(r_i/2)^2}
       \leq 24j.
\end{aligned}
\end{equation}

Let $I(v,w)$ be the interference at $v$ caused by $w$. Denote by $I_{out}$ the interference caused by the nodes outside $D_{r_i}(v)$. Then one can calculate $I_{out}$ as follows:
\begin{equation}
\begin{aligned}
    I_{out} &= \sum_{j=1}^{\infty} \sum_{w\in S_j} I(v,w)
            = \sum_{j=1}^{\infty} \sum_{w\in S_j} \frac{P_i}{d(v,w)^{\alpha}} \cdot p\\
           &\leq \sum_{j=1}^{\infty} |S_j|\frac{1}{\sigma}\cdot \frac{2\mathcal{N}\beta r_i^{\alpha}}{(jr_i)^{\alpha}}
           \leq \frac{24\beta(\alpha-1)}{\sigma(\alpha-2)}\cdot N
           \leq N/2,
\end{aligned}
\end{equation}
where the last inequality holds when $\sigma>\frac{48\beta(\alpha-1)}{(\alpha-2)}$. Considering the case when $u$ is the only node that transmits in the current slot, since $u\in V_{i-1} \cap N_{r_i}(v)$, $d(v,u)\leq r_i$, we have
\begin{equation}
\begin{aligned}
    SINR(v,u)= 
    \frac{
        \frac{P_i}{d(v,u)^{\alpha}}
    }{
        N+I_{out}
    }
    \geq 
    \frac{
        \frac{2\mathcal{N}\beta r_i^{\alpha}}{r_i^{\alpha}}
    }{
        N+N/2
    }
    \geq \beta,
\end{aligned}
\end{equation}
which indicates that if $u$ is the only node who transmits, $v$ can receive the message. Then we bound the probability that $u$ is the only transmitting node. Since $P_i(v)\leq \frac{1}{\sigma}$, the probability that only $u$ transmits at each slot is
\begin{equation}
\begin{aligned}
p_u \prod_{w\in V_i \cap N_{r_i}(v)\backslash u}(1-p_w) &\geq p_u\prod_{w\in V_i \cap N_{r_i}(v)}(1-p_w)\\
&\geq p_u\prod_{w\in V_i \cap N_{r_i}(v)} e^{\frac{-p_w}{1-p}}\\
&= p e^{\frac{-P_i(v)}{1-p}}\\
&\geq (\sigma\lambda')^{-1} e^{\frac{-\lambda'}{\sigma\lambda'-1}}\\
&\in \Omega(1).
\end{aligned}
\end{equation}

This implies that with a constant probability, $v$ can receive $u$'s message. 
\end{proof}


In Algorithm~\ref{alg:data:agg}, each round consists of a fixed number of $\mu \cdot \log N$ slots. At each slot, a child $u\in V_{i-1}\setminus V_i$ transmits constantly with probability $p=\frac{1}{\sigma\lambda'}$. Lemma~\ref{lemma:cluster:com} indicates that at each slot, $u$ can succeed in sending a message to its parent $v$ with a constant probability denoted by $\hat{p}$. Thus, by applying the Chernoff bound (see Lemma.~\ref{lemma:chernoff} in Sec.~\ref{sub:chernoff}), the probability that $u$ succeeds in sending a message to its parent after $\mu\cdot \log N$ slots is $1-(1-\hat{p})^{\mu\log N}\geq 1-e^{-\hat{p}\mu\log N}\geq 1-N^{-2}$ if $\mu\geq 2/\hat{p}$. Since the density of the active nodes is bounded by $\lambda'=25$, the probability that all children succeed is $(1-N^2)^{\lambda'}$. 

Next, assume that the nodes have synchronized clocks. Then at each round $R_i$, the nodes in different independent sets can send messages to their parents at the same time. The probability that the data has been aggregated to all parent nodes in $R_i$ is at least $(1-N^2)^{\lambda'V_{i+1}}\geq(1-N^{-1})$ since $\lambda'V_{i+1}<N$. Thus, the one-round data aggregation succeeds in $O(\log N)$ w.h.p. Considering the $(\log\Gamma)$-round aggregation process, one can immediately derive that data aggregation succeeds in $O(\log N\log\Gamma)$ w.h.p. 

Unlike a normal data aggregation, Algorithm~\ref{alg:reaggregation} terminates when the messages from all non-faulty nodes are received by the collector without any loss. A message can be lost when crash failure happens. The number of faulty nodes is bounded by $f$, thus during the data reaggregation process, there are at most $f$ times of the execution of the spanner construction and data aggregation, which gives the upper bound of the runtime as $O(f\log N \log\Gamma)$.
\end{proof}

\subsection{Persistence and Liveness}

In this subsection, we demonstrate how our protocol ensures persistence and liveness properties whose definitions are adapted from the rigorous ones proposed by Garay \textit{et al.} \cite{garay2015bitcoin}.

\begin{theorem}{\textbf{Persistence.}}
\label{thm:persistence}
   If a non-faulty node $v$ proclaims a transaction $tx_v$ in the position $tx^j_i$, other nodes, if queried, should report the same result. 
\end{theorem}
\begin{proof}
To prove the persistence property, we need to show that for any two blockchains $BC_v$ and $BC_u$ of nodes $v$ and $u$, respectively, one cannot find two different transactions $tx_v\in BC_v$ and $tx_u\in BC_u$ that are in the same position $tx_i^j$. To prove by contradiction, we assume that such $tx_v$ and $tx_u$ exist, and there are two cases when the assumption can hold. 

C1: $tx_v$ and $tx_u$ are respectively appended to blockchains $BC_v$ and $BC_u$ at the same epoch. This indicates that a leader broadcasts two different blocks in the same epoch, which is not permissible in $\mathit{wChain}$, thus contradicting our assumption. 

C2: $tx_v$ and $tx_u$ are appended to their corresponding blockchains $BC_v$ and $BC_u$ in two different epochs $e_m$ and $e_n$. Let $tx_i^j$ also denote the transaction generated for the first time in position $tx_i^j$ and appended to the blockchains of at least $f+1$ nodes in $e_i$. Since a leader cannot broadcast two different blocks in the same epoch, the nodes who append $tx_i^j$ to their local blockchain in $e_i$ should have an identical view of $tx_i^j$. Without loss of generality, assume $i<m<n$. Using contradiction, we assume $tx_v\neq tx_i^j$. Since $i<m$, $v$ should crash before $e_m$ and recover in $e_m$ so that $tx_v$ is appended to $BC_v$ when $v$ updates its blockchain by applying $update()$. The leader who sends the update information in $e_m$ has an identical view with at least $f$ nodes, which means that at least $f+1$ nodes have the same view on $tx_i^j$ in $e_m$. Since there are also at least $f+s$ nodes who agree on $tx_i^j$ in $e_i$ and the network size $N=2f+1$, we have a contradiction saying that $N>2f+1+s>N$. That is, one can only have $tx_i^j=tx_v$. By applying the same proof, we can obtain $tx_i^j=tx_u$. Hence $tx_i^j=tx_v=tx_u$, which contradicts the assumption that $tx_v$ and $tx_u$ are different. 

In a nutshell, all the nodes queried for a transaction in a specific position should report the same result or report error messages. 
\end{proof}

\begin{theorem}{\textbf{Liveness.}}
\label{thm:liveness}
  If a non-faulty node generates a transaction and contends to send it, $\mathit{wChain}$ can add it to the blockchains within $T$ slots w.h.p., where the upper bound of $T$ is $O(\log N \log\Gamma)$ when crash failures happen in a low frequency, and the worst-case upper bound of $T$ is $O(f\log N \log\Gamma)$.
\end{theorem}

\begin{proof}
In a specific epoch, the best case for $\mathit{wChain}$ occurs when no crash failures happen for all nodes throughout the epoch so that the leader executing the $Reaggregation(data_{\ell})$ subroutine can broadcast $m_{\ell}\leftarrow\textit{MSG}(stop_{\ell})$ without the need of waiting for the aggregated data from a new collector. By Theorem~\ref{thm:time}, the view messages and the transactions can be fully aggregated within $O(\log N \log\Gamma)$ w.h.p. Besides, the  \texttt{DECIDE} phase takes $O(1)$ slots. Hence the upper bound of $T$ is $O(\log N \log\Gamma)$. In a normal case, failures happen in a low frequency so that data aggregation can be executed in $O(1)$ time during the data reaggregation process, and the \texttt{PREPARE} and  \texttt{COMMIT} phases take $O(\log N\log\Gamma)$ slots in total. Therefore, the upper bound of $T$ is still $O(\log N \log\Gamma)$. 

From the worst-case perspective, we assume that before epoch $e_i$, nodes are all non-faulty, and $f$ nodes crash during $e_i$. If a leader crashes, the worst-case runtime of the current epoch is still $O(\log N\log\Gamma)$ when it crashes in the  \texttt{DECIDE} phase. If a follower $v$ crashes, some data from $v$'s children cannot be collected to the leader such that the leader must execute $SpannerConstruction$ and $DataAggregation(data_v)$, which take $O(\log N\log\Gamma)$ slots. Thus the extra runtime brought by a one-time crash failure is $O(\log N\log\Gamma)$. If $f$ nodes crash, the total time spent on waiting for a successful  \texttt{DECIDE} phase is bounded by $O(f\log N\log\Gamma)$. This gives the normal and worst-case upper bound of $T$ as $O(\log N \log\Gamma)$ and $O(f\log N \log\Gamma)$, respectively.
\end{proof}

\section{Simulation Results}
\label{sec:simulation}

\begin{figure*}[!htb]
\centering
\subfigure[$150\times150$ plane]{
    \label{fig:el_n_uge}
    \centering
    \includegraphics[width=2.7in]{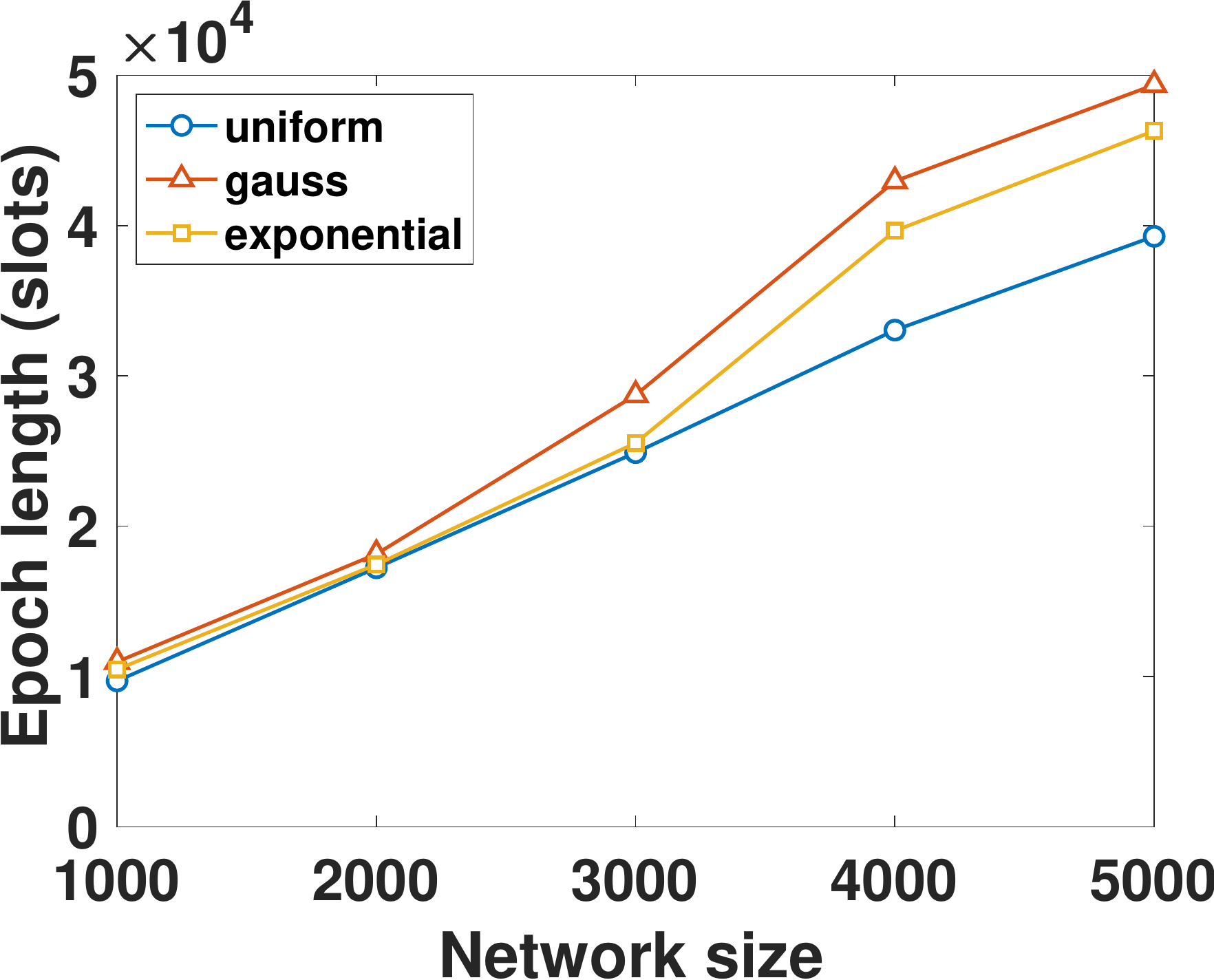}
}%
\hspace{.4in}
\subfigure[$150\times150$ plane]{
    \label{fig:tps_n_uge}
    \centering
    \includegraphics[width=2.8in]{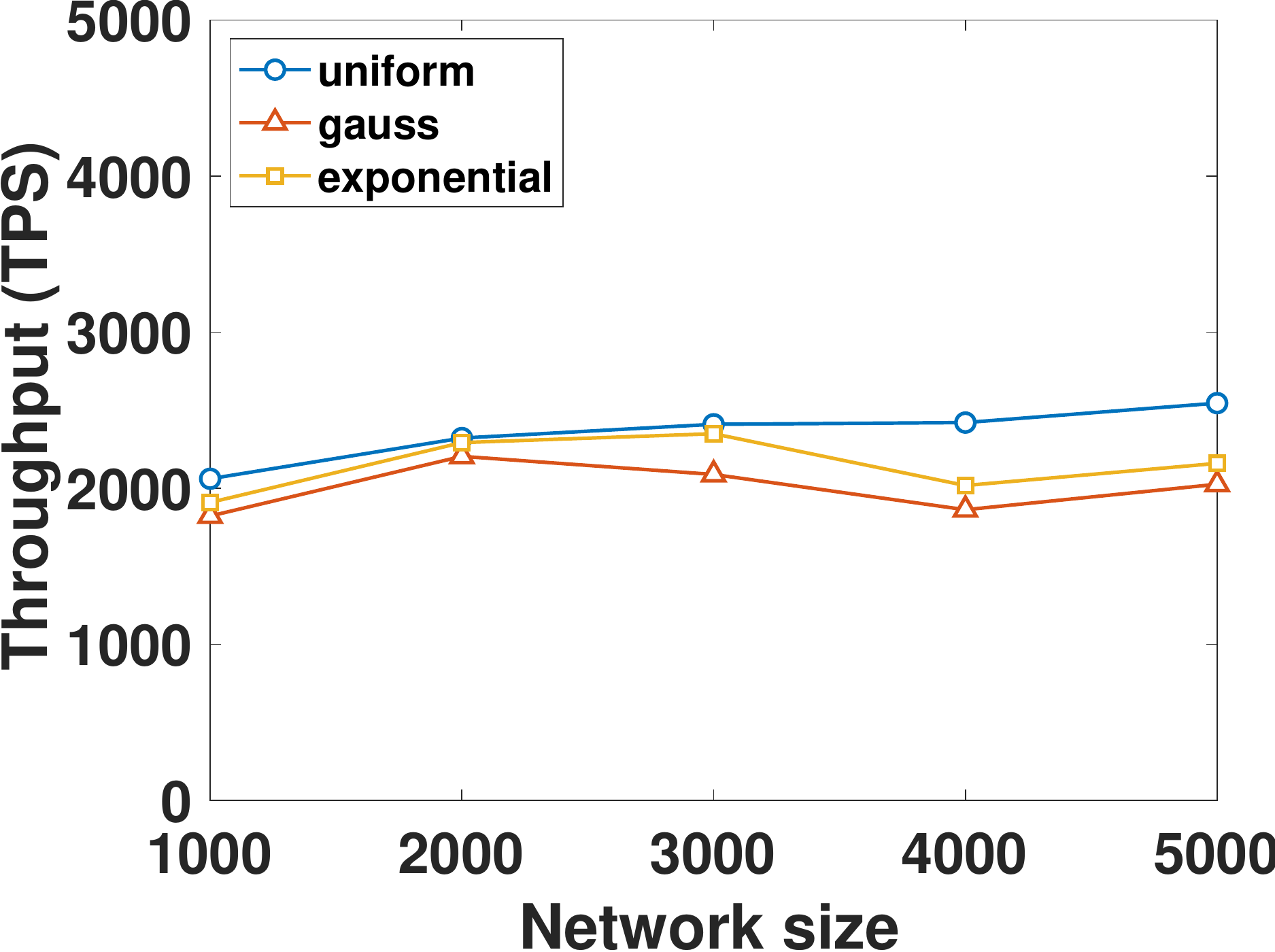}
}%

\subfigure[$N=2000$]{
    \label{fig:el_gamma_uge}
    \centering
    \includegraphics[width=2.7in]{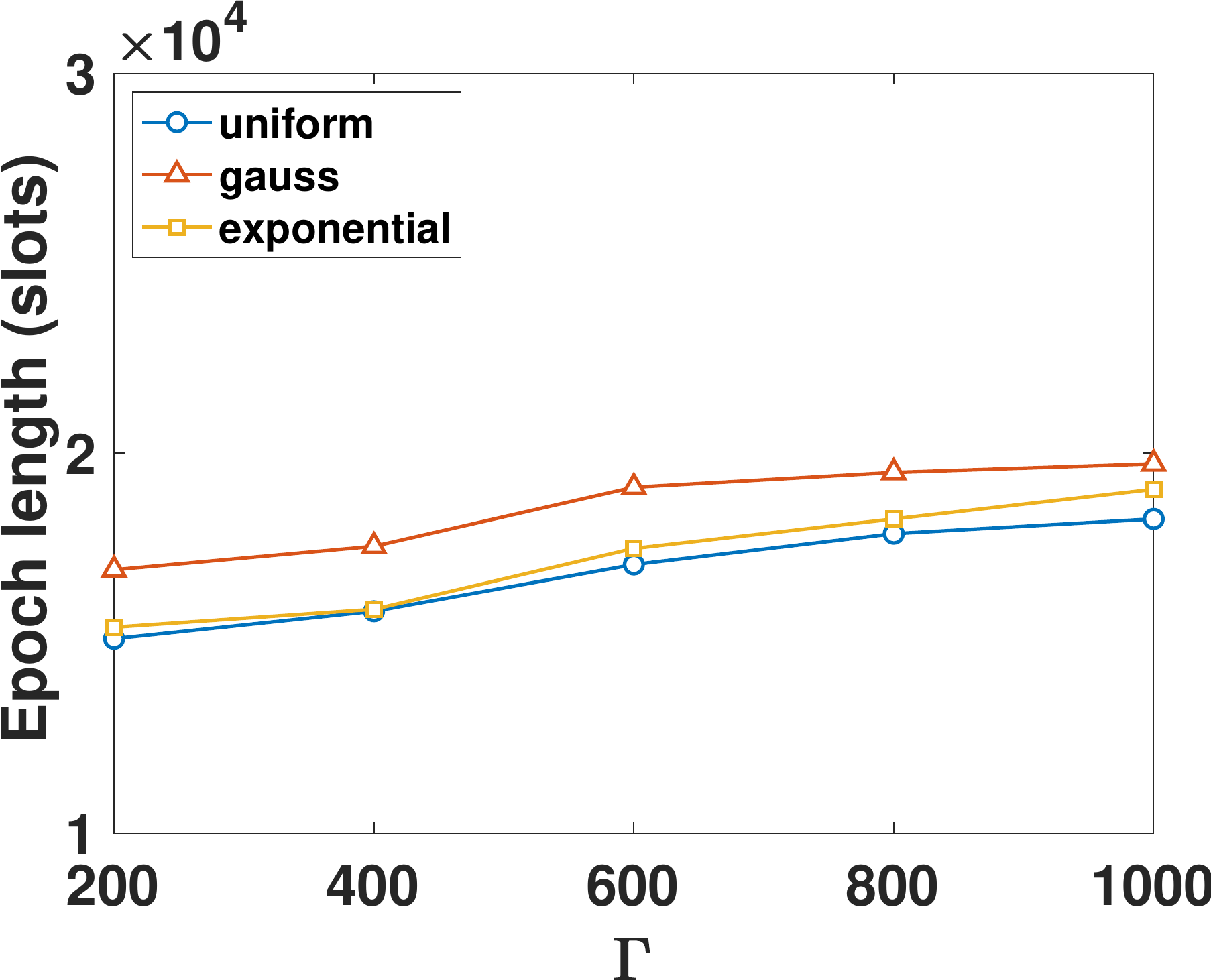}
}%
\hspace{.4in}
\subfigure[$N=2000$]{
    \label{fig:tps_gamma_uge}
    \centering
    \includegraphics[width=2.8in]{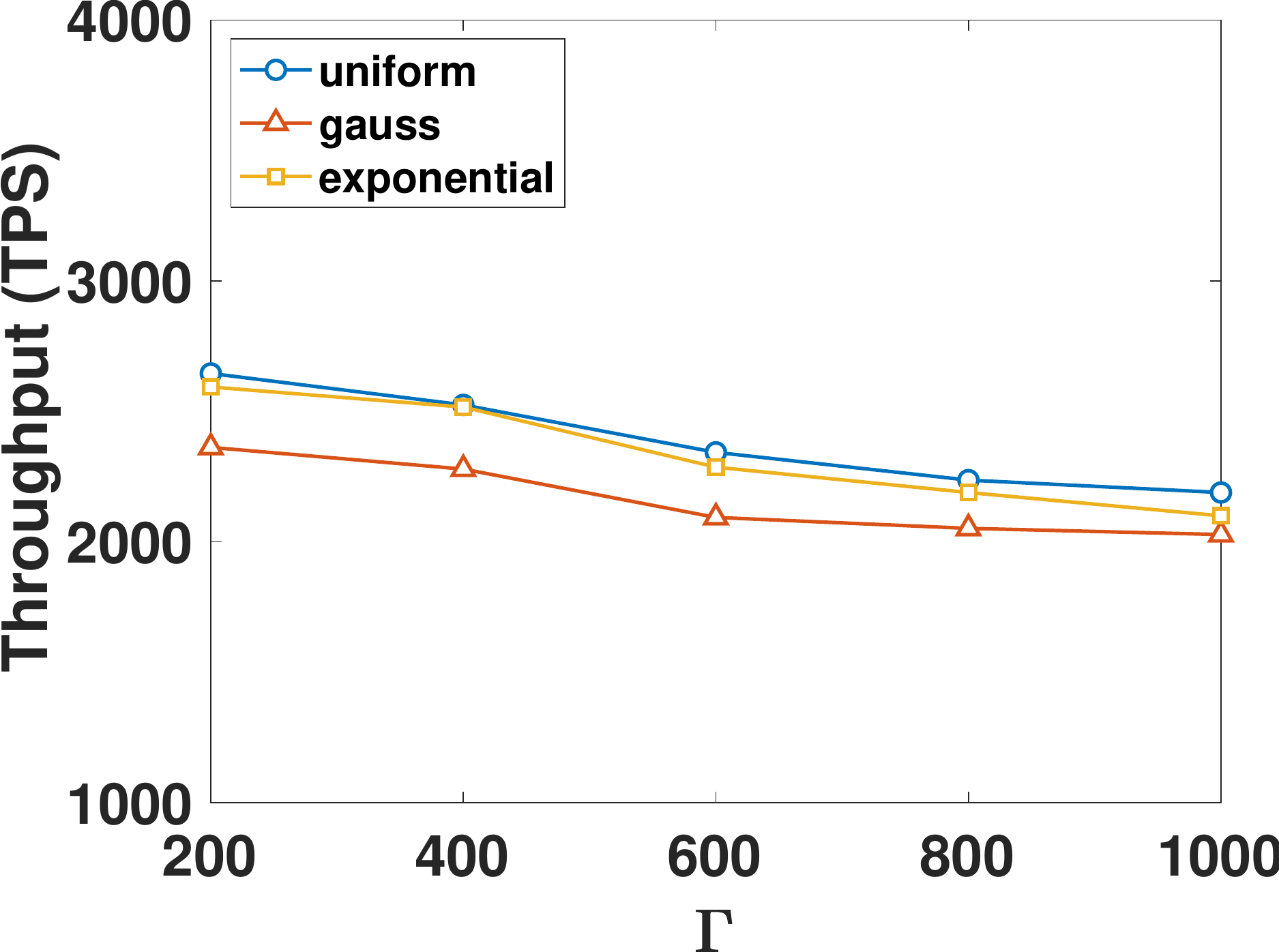}
}%
\caption{The performance of $\mathit{wChain}$ \emph{vs.} the network size $N$ and $\Gamma$ (under uniform, normal, and exponential distributions).}
\label{fig:N:gamma}
\end{figure*}

In this section, we conduct simulation experiments to validate the performance of $\mathit{wChain}$. The impacts of various parameters are investigated, including the SINR model parameters, network size $N$, and $\Gamma$, the ratio of the maximum distance to the minimum distance between nodes. If not stated otherwise, we adopt the following parameter settings: $\alpha=3$, $\beta=3, s=100, \lambda'=25, \mathcal{N}=1, P_i=2\mathcal{N}\beta 2^{i\alpha}$. The frequency of node crashes is set to be $1\%\times N$ nodes per second. To evaluate the performance, we adopt two metrics, namely \emph{epoch length} and \emph{throughput}. The epoch length is the number of slots within an epoch; given that the unit slot time for IEEE 802.11 is set to be $50\mu s$, one can calculate throughput as
\begin{equation}
\text{Throughput} = \frac{\textit{The number of transactions}}{\textit{Epoch length}\times 50\mu s}, 
\end{equation}
hence the unit of throughput is transactions per second (TPS). 

The simulation program is written in C and all the experiments are performed under a CentOS 7 operating system running on a machine with an Intel Xeon 3.4 GHz CPU, 120 GB RAM, and 1 TB SATA Hard Drive. Over 20 runs are carried out to get the average for each result. 

\begin{figure*}[!htbp]
    \centering
    \subfigure[$150\times150$ plane]{
        \label{fig:el_n_ab}
        \centering
        \includegraphics[width=2.7in]{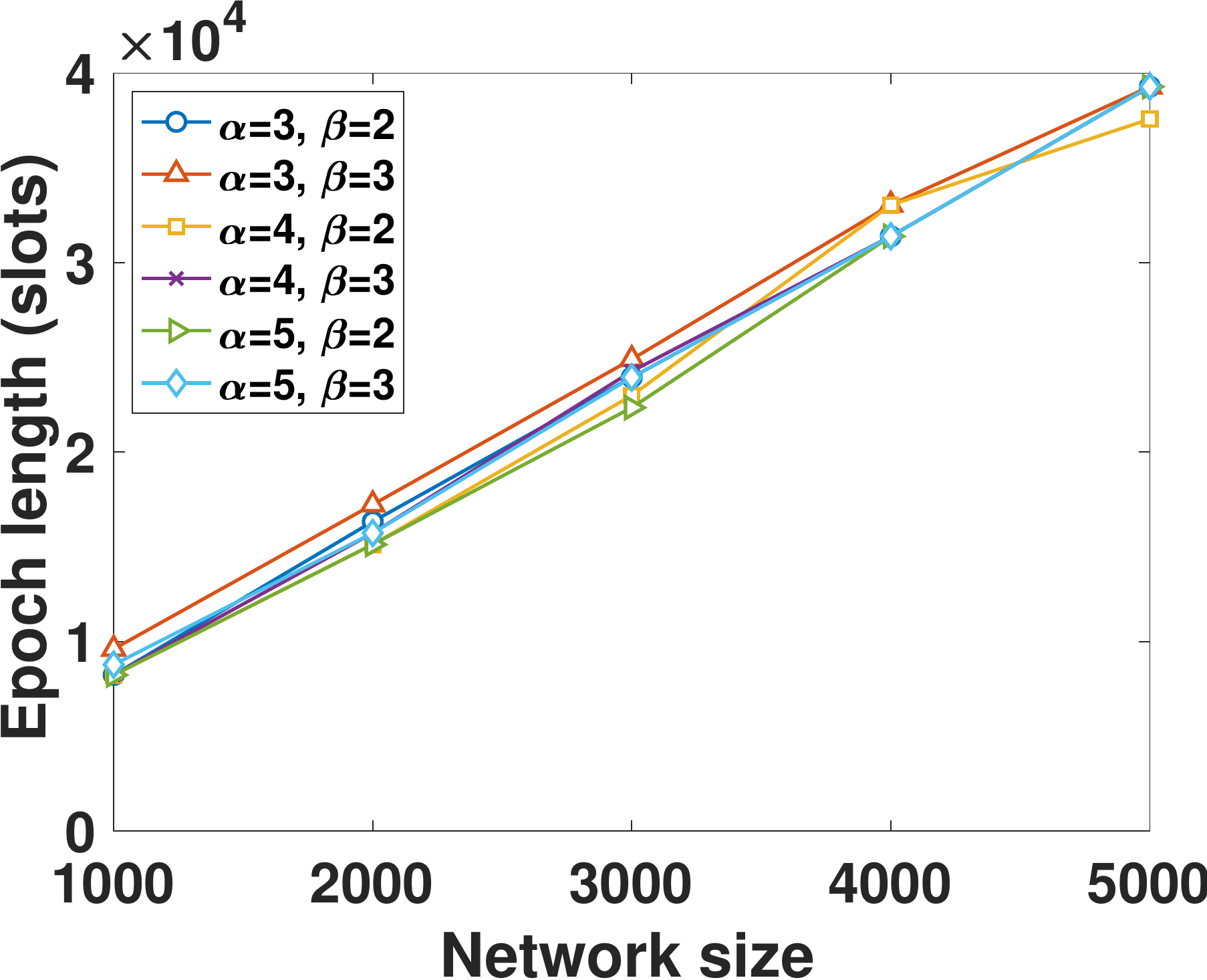}
    }%
    \hspace{.4in}
    \subfigure[$150\times150$ plane]{
        \label{fig:tps_n_ab}
        \centering
        \includegraphics[width=2.8in]{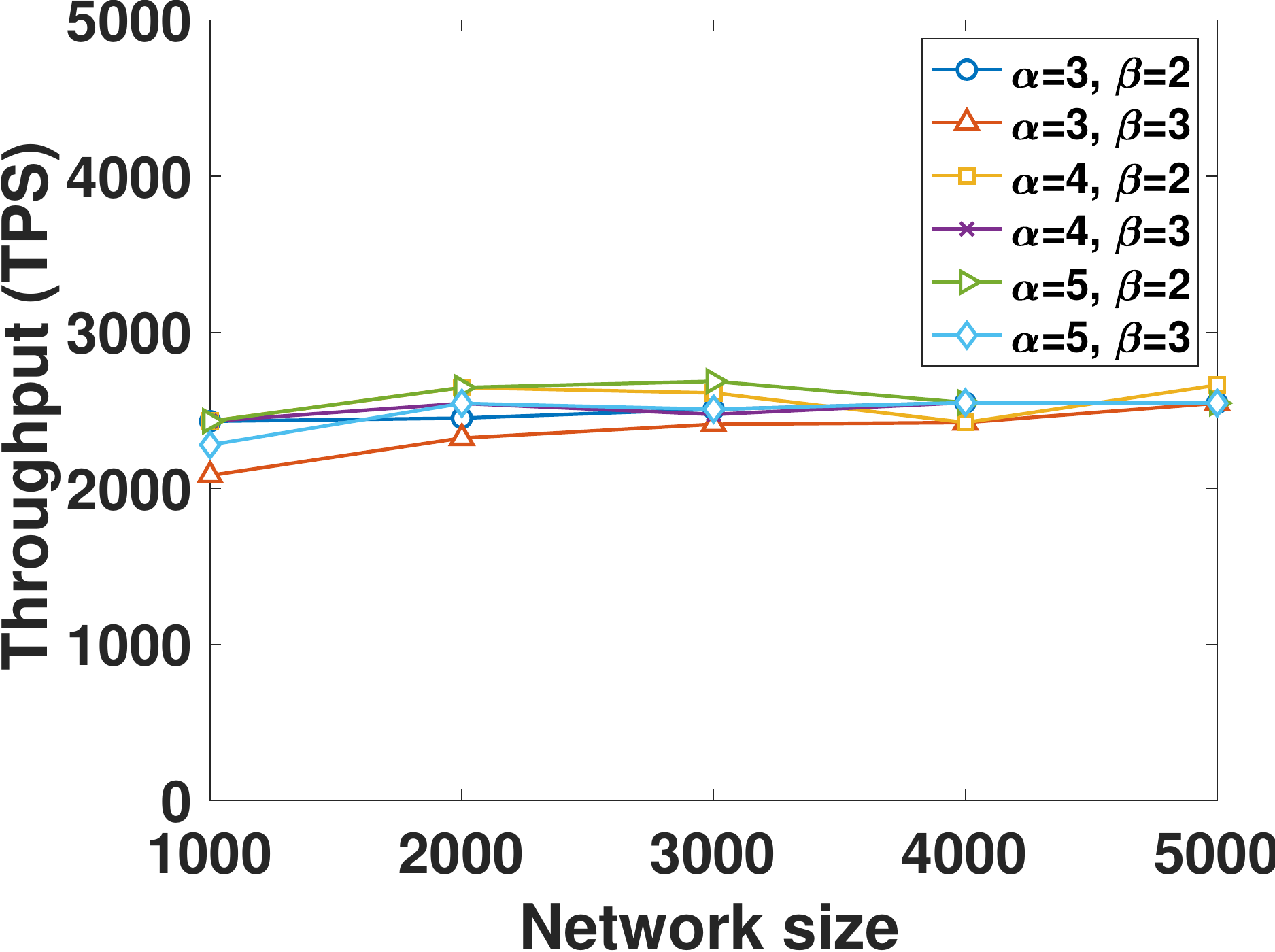}
    }%

    \subfigure[$N=2000$]{
        \label{fig:el_gamma_ab}
        \centering
        \includegraphics[width=2.7in]{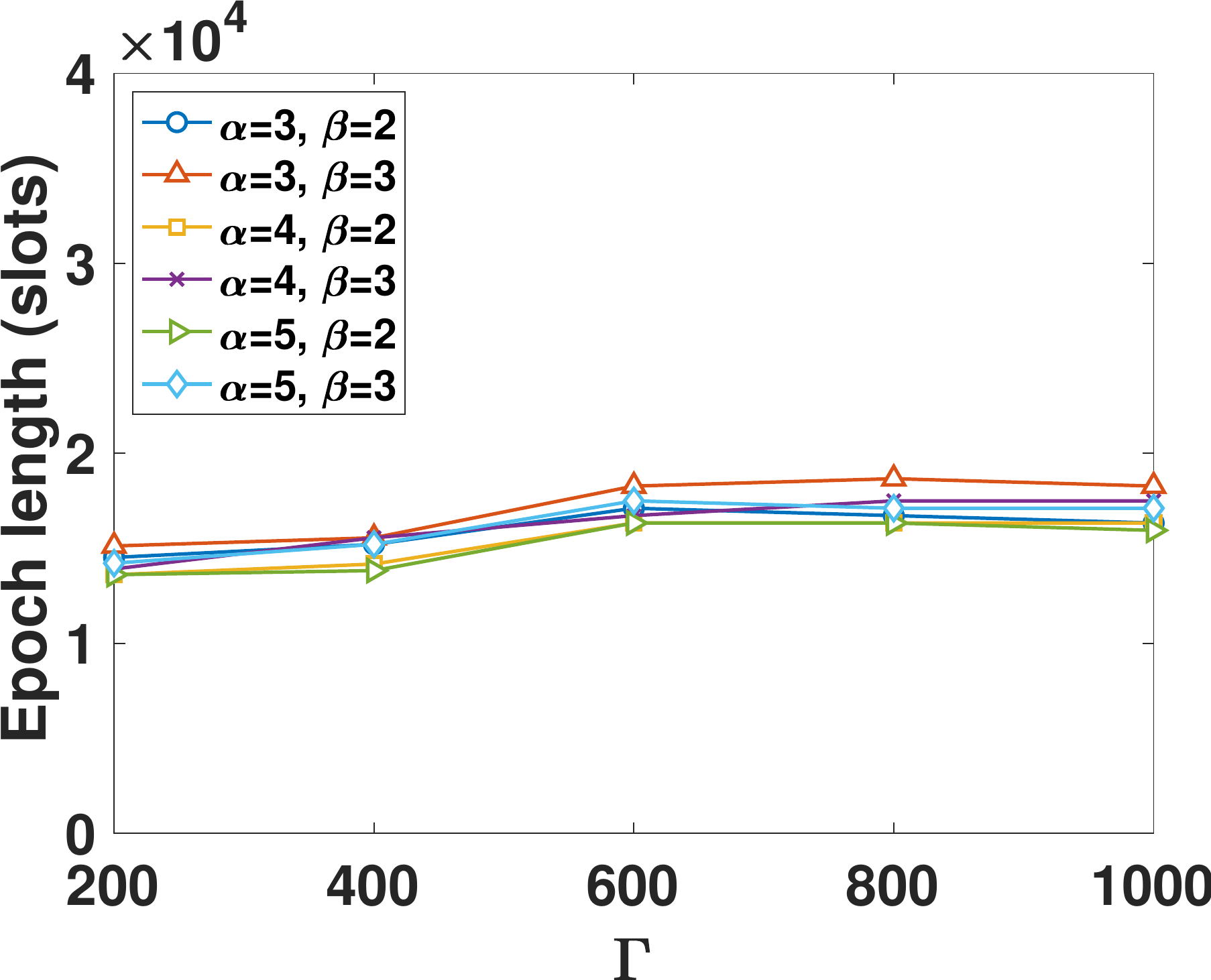}
    }%
    \hspace{.4in}
    \subfigure[$N=2000$]{
        \label{fig:tps_gamma_ab}
        \centering
        \includegraphics[width=2.8in]{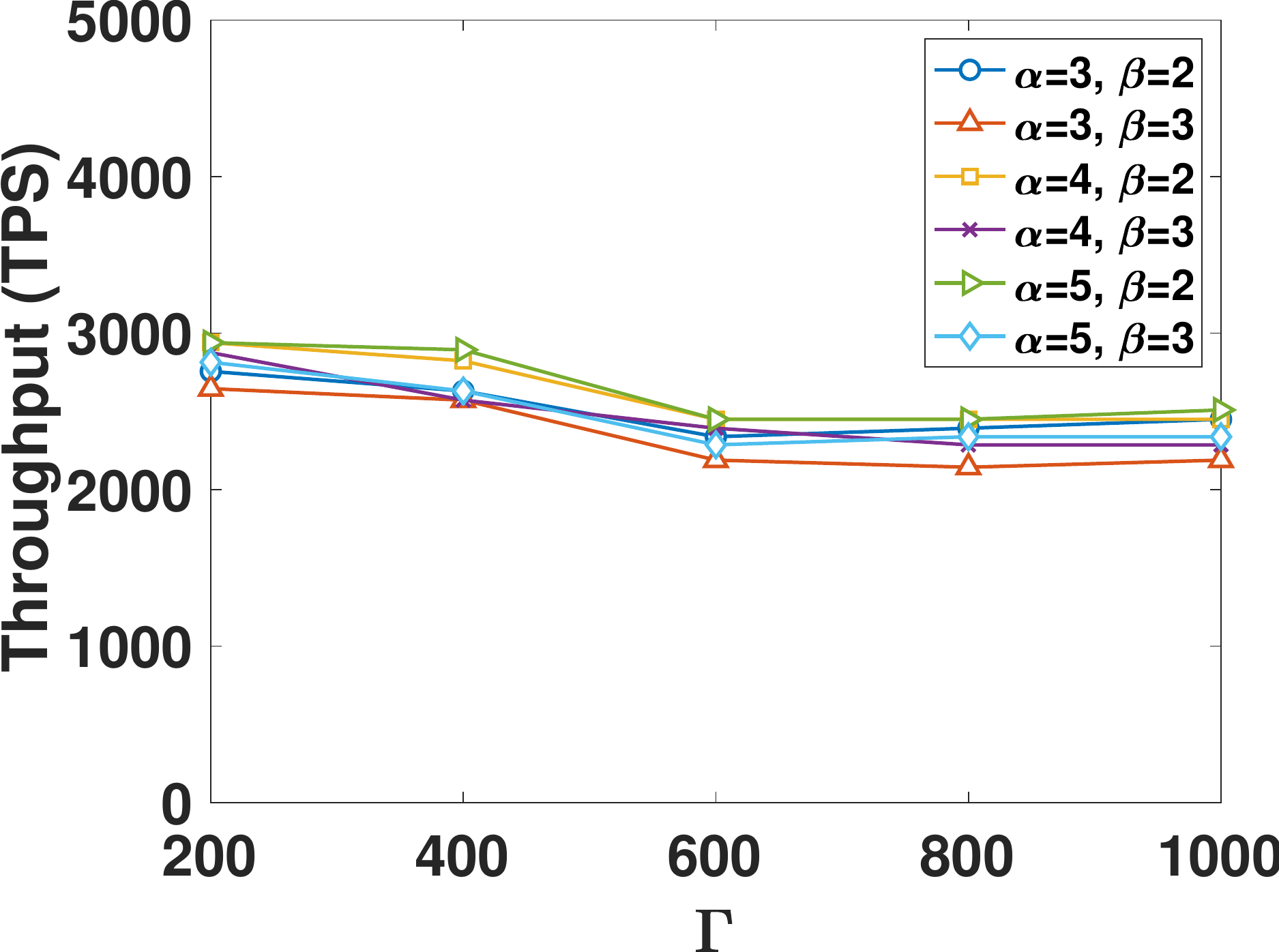}
    }%
    \caption{The performance of $\mathit{wChain}$ with various $\alpha$ and $\beta$ (under a uniform distribution).}
\end{figure*}

\subsection{Impacts of Network Size and $\Gamma$}

Since liveness is mainly determined by $N$ and $\Gamma$, we first study the impacts of $N$ and $\Gamma$ on the performance of $\mathit{wChain}$. We consider three types of distributions, namely uniform, normal, and exponential. To investigate the impacts of $N$, we adopt parameters $\alpha=3, \beta=3$ for the SINR model, and the plane is of size $150\times150$. The results are reported in Fig.~\ref{fig:N:gamma}.

One can observe  from Fig.~\ref{fig:el_n_uge} that the epoch length increases with $N$. With a uniform distribution and $N=5000$, it is 49364 slots (about 2.47s). Under normal and exponential distributions, the nodes have larger epoch lengths since they are denser in the center or the corner. They may suffer from heavier contention and spend more time transmitting a message. This result is consistent with our model assumption which states that the network density should be limited. 

Fig.~\ref{fig:tps_n_uge} indicates that under uniform distributions $\mathit{wChain}$ has the highest throughput, which only increases with $N$. When $N=5000$, the throughput with uniform distributions reaches 2546 TPS and is about 28\% higher than that for normal distributions. Under normal distributions, the throughput can reach 1986 TPS when $N=5000$. Note that that the throughput under normal or exponential distributions has a small decrease from $N=3000$. This is because the density is so high in some areas with a large number of nodes that contend heavily, negatively affecting throughput. 

Then we investigate the impact of $\Gamma$ and set $N=2000$. In Fig.~\ref{fig:el_gamma_uge}, with an increasing $\Gamma$, the epoch length increases since the spanner has more levels. The epoch lengths under normal and exponential distributions are still larger than the one under uniform distributions, which is caused by the same reason as Fig.~\ref{fig:el_n_uge}. 

Concerning that the throughput is a function of $\Gamma$, our protocol running under uniform distributions yields the largest throughput. However, the throughput decreases with a larger $\Gamma$ because a larger $\Gamma$ indicates that the spanner has more levels, and the data aggregation process takes a longer time. 

\subsection{Impacts of the SINR model parameters}

We perform four experiments to explore how the SINR model parameters, namely $\alpha$ and $\beta$, impact the performance of $\mathit{wChain}$. Assume that nodes are uniformly distributed in the plane. The combinations of $\alpha=3, 4, 5$ and $\beta=2, 3$ are tested with different $N$ and $\Gamma$ in our experiments. 

As shown in Fig.~\ref{fig:el_n_ab}, the epoch lengths scale from about $1\times10^4$ to $4\times10^4$ slots with $N$ under different settings of $\alpha$ and $\beta$ but they differ very little for different $\alpha$ and $\beta$ and the same $N$. 
In \ref{fig:tps_n_ab}, even though we observe a small throughput decrease when $\alpha=3$ and $\beta=3$, the epoch length and throughput are independent of the values of $\alpha$ and $\beta$ with different $N$. 
In Fig.~\ref{fig:el_gamma_ab} and \ref{fig:tps_gamma_ab}, the epoch length and throughput slightly change with different $\alpha$ and $\beta$. This is because $\alpha$, as the path-loss exponent, works closely with the distance between nodes as well as $\Gamma$. The impacts of $\alpha$ and $\beta$ on throughput are in an allowable range so that one can claim that the performance of $\mathit{wChain}$ is insensitive to them.

\section{Conclusion and Future Research}
\label{sec:discussion}
In this paper, we propose a fast fault-tolerant blockchain protocol, namely $\mathit{wChain}$, which can ensure the fast data aggregation leveraging a spanner as the communication backbone. The runtime upper bound of our protocol is $O(\log N\log\Gamma)$ when crash failures happen in a low frequency, and the worst-case upper bound of $\mathit{wChain}$ is $O(f\log N\log\Gamma)$. Besides, $\mathit{wChain}$ tolerates at most $f=\lfloor \frac{N}{2} \rfloor$ faulty nodes and is capable of handling node recovery while satisfying persistence and liveness, the two crucial properties for a blockchain protocol to function well. Both theoretical analysis and simulation studies are conducted to validate our design. On the last point, we only grapple with crash failures in this paper, but it would be interesting to consider Byzantine fault-tolerance in wireless networks. It is also worthy of investigating the impacts of mobility in ad hoc wireless networks.

\section*{Acknowledgment}
This study was partially supported by the National Natural Science Foundation of China under Grants 61771289, 61871466, 61832012, and 61672321, the Key Science and Technology Project of Guangxi under Grant AB19110044, and the US National Science Foundation under grants IIS-1741279 and CNS-1704397.

\bibliographystyle{IEEEtran}
\bibliography{references}

\begin{thebibliography}{10}
\providecommand{\url}[1]{#1}
\csname url@samestyle\endcsname
\providecommand{\newblock}{\relax}
\providecommand{\bibinfo}[2]{#2}
\providecommand{\BIBentrySTDinterwordspacing}{\spaceskip=0pt\relax}
\providecommand{\BIBentryALTinterwordstretchfactor}{4}
\providecommand{\BIBentryALTinterwordspacing}{\spaceskip=\fontdimen2\font plus
\BIBentryALTinterwordstretchfactor\fontdimen3\font minus
  \fontdimen4\font\relax}
\providecommand{\BIBforeignlanguage}[2]{{%
\expandafter\ifx\csname l@#1\endcsname\relax
\typeout{** WARNING: IEEEtran.bst: No hyphenation pattern has been}%
\typeout{** loaded for the language `#1'. Using the pattern for}%
\typeout{** the default language instead.}%
\else
\language=\csname l@#1\endcsname
\fi
#2}}
\providecommand{\BIBdecl}{\relax}
\BIBdecl

\bibitem{feng2020joint}
J.~Feng, F.~R. Yu, Q.~Pei, J.~Du, and L.~Zhu, ``Joint optimization of radio and
  computational resources allocation in blockchain-enabled mobile edge
  computing systems,'' \emph{IEEE Transactions on Wireless Communications},
  2020.

\bibitem{DBLP:journals/network/DaiXMCHZ19}
\BIBentryALTinterwordspacing
Y.~Dai, D.~Xu, S.~Maharjan, Z.~Chen, Q.~He, and Y.~Zhang, ``Blockchain and deep
  reinforcement learning empowered intelligent 5g beyond,'' \emph{{IEEE}
  Network}, vol.~33, no.~3, pp. 10--17, 2019. [Online]. Available:
  \url{https://doi.org/10.1109/MNET.2019.1800376}
\BIBentrySTDinterwordspacing

\bibitem{DBLP:journals/winet/MalikNHL20}
\BIBentryALTinterwordspacing
N.~Malik, P.~Nanda, X.~He, and R.~P. Liu, ``Vehicular networks with security
  and trust management solutions: proposed secured message exchange via
  blockchain technology,'' \emph{Wirel. Networks}, vol.~26, no.~6, pp.
  4207--4226, 2020. [Online]. Available:
  \url{https://doi.org/10.1007/s11276-020-02325-z}
\BIBentrySTDinterwordspacing

\bibitem{8936389}
T.~{Kim}, R.~{Goyat}, M.~K. {Rai}, G.~{Kumar}, W.~J. {Buchanan}, R.~{Saha}, and
  R.~{Thomas}, ``A novel trust evaluation process for secure localization using
  a decentralized blockchain in wireless sensor networks,'' \emph{IEEE Access},
  vol.~7, pp. 184\,133--184\,144, 2019.

\bibitem{pbft}
M.~Castro, B.~Liskov \emph{et~al.}, ``Practical byzantine fault tolerance,'' in
  \emph{OSDI}, vol.~99, no. 1999, 1999, pp. 173--186.

\bibitem{kwon2014tendermint}
J.~Kwon, ``Tendermint: Consensus without mining,'' \emph{Draft v. 0.6, fall},
  vol.~1, p.~11, 2014.

\bibitem{yin2019hotstuff}
M.~Yin, D.~Malkhi, M.~K. Reiter, G.~G. Gueta, and I.~Abraham, ``Hotstuff: Bft
  consensus with linearity and responsiveness,'' in \emph{Proceedings of the
  2019 ACM Symposium on Principles of Distributed Computing}, 2019, pp.
  347--356.

\bibitem{poirot2019paxos}
V.~Poirot, B.~Al~Nahas, and O.~Landsiedel, ``Paxos made wireless: Consensus in
  the air.'' in \emph{EWSN}, 2019, pp. 1--12.

\bibitem{kuhn2009abstract}
F.~Kuhn, N.~Lynch, and C.~Newport, ``The abstract mac layer,'' in
  \emph{International Symposium on Distributed Computing}.\hskip 1em plus 0.5em
  minus 0.4em\relax Springer, 2009, pp. 48--62.

\bibitem{yu2018exact}
D.~Yu, Y.~Zhang, Y.~Huang, H.~Jin, J.~Yu, and Q.-S. Hua, ``Exact implementation
  of abstract mac layer via carrier sensing,'' in \emph{IEEE INFOCOM 2018-IEEE
  Conference on Computer Communications}.\hskip 1em plus 0.5em minus
  0.4em\relax IEEE, 2018, pp. 1196--1204.

\bibitem{liu2020deep}
M.~Liu, Y.~Teng, F.~R. Yu, V.~C. Leung, and M.~Song, ``A deep reinforcement
  learning-based transcoder selection framework for blockchain-enabled wireless
  d2d transcoding,'' \emph{IEEE Transactions on Communications}, 2020.

\bibitem{onireti2019viable}
O.~Onireti, L.~Zhang, and M.~A. Imran, ``On the viable area of wireless
  practical byzantine fault tolerance (pbft) blockchain networks,'' in
  \emph{2019 IEEE Global Communications Conference (GLOBECOM)}.\hskip 1em plus
  0.5em minus 0.4em\relax IEEE, 2019, pp. 1--6.

\bibitem{ren2018incentive}
Y.~Ren, Y.~Liu, S.~Ji, A.~K. Sangaiah, and J.~Wang, ``Incentive mechanism of
  data storage based on blockchain for wireless sensor networks,'' \emph{Mobile
  Information Systems}, vol. 2018, 2018.

\bibitem{liu2018joint}
M.~Liu, F.~R. Yu, Y.~Teng, V.~C. Leung, and M.~Song, ``Joint computation
  offloading and content caching for wireless blockchain networks,'' in
  \emph{IEEE INFOCOM 2018-IEEE Conference on Computer Communications Workshops
  (INFOCOM WKSHPS)}.\hskip 1em plus 0.5em minus 0.4em\relax IEEE, 2018, pp.
  517--522.

\bibitem{sun2019blockchain}
Y.~Sun, L.~Zhang, G.~Feng, B.~Yang, B.~Cao, and M.~A. Imran,
  ``Blockchain-enabled wireless internet of things: Performance analysis and
  optimal communication node deployment,'' \emph{IEEE Internet of Things
  Journal}, vol.~6, no.~3, pp. 5791--5802, 2019.

\bibitem{moniz2012Byzantine}
H.~Moniz, N.~F. Neves, and M.~Correia, ``Byzantine fault-tolerant consensus in
  wireless ad hoc networks,'' \emph{IEEE Transactions on Mobile Computing},
  vol.~12, no.~12, pp. 2441--2454, 2012.

\bibitem{newport2014consensus}
C.~Newport, ``Consensus with an abstract mac layer,'' in \emph{Proceedings of
  the 2014 ACM symposium on Principles of distributed computing}, 2014, pp.
  66--75.

\bibitem{newport2018fault}
C.~Newport and P.~Robinson, ``Fault-tolerant consensus with an abstract mac
  layer,'' \emph{arXiv preprint arXiv:1810.02848}, 2018.

\bibitem{dong2009resilient}
Q.~Dong and D.~Liu, ``Resilient cluster leader election for wireless sensor
  networks,'' in \emph{2009 6th Annual IEEE Communications Society Conference
  on Sensor, Mesh and Ad Hoc Communications and Networks}.\hskip 1em plus 0.5em
  minus 0.4em\relax IEEE, 2009, pp. 1--9.

\bibitem{vasudevan2003leader}
S.~Vasudevan, B.~DeCleene, N.~Immerman, J.~Kurose, and D.~Towsley, ``Leader
  election algorithms for wireless ad hoc networks,'' in \emph{Proceedings
  DARPA Information Survivability Conference and Exposition}, vol.~1.\hskip 1em
  plus 0.5em minus 0.4em\relax IEEE, 2003, pp. 261--272.

\bibitem{raychoudhury2008top}
V.~Raychoudhury, J.~Cao, and W.~Wu, ``Top k-leader election in wireless ad hoc
  networks,'' in \emph{2008 Proceedings of 17th International Conference on
  Computer Communications and Networks}.\hskip 1em plus 0.5em minus 0.4em\relax
  IEEE, 2008, pp. 1--6.

\bibitem{chockler2005consensus}
G.~Chockler, M.~Demirbas, S.~Gilbert, C.~Newport, and T.~Nolte, ``Consensus and
  collision detectors in wireless ad hoc networks,'' in \emph{Proceedings of
  the twenty-fourth annual ACM symposium on Principles of distributed
  computing}, 2005, pp. 197--206.

\bibitem{scutari2008distributed}
G.~Scutari and S.~Barbarossa, ``Distributed consensus over wireless sensor
  networks affected by multipath fading,'' \emph{IEEE Transactions on Signal
  Processing}, vol.~56, no.~8, pp. 4100--4106, 2008.

\bibitem{aysal2009reaching}
T.~C. Aysal, A.~D. Sarwate, and A.~G. Dimakis, ``Reaching consensus in wireless
  networks with probabilistic broadcast,'' in \emph{2009 47th Annual Allerton
  Conference on Communication, Control, and Computing (Allerton)}.\hskip 1em
  plus 0.5em minus 0.4em\relax IEEE, 2009, pp. 732--739.

\bibitem{richa2011self}
A.~Richa, C.~Scheideler, S.~Schmid, and J.~Zhang, ``Self-stabilizing leader
  election for single-hop wireless networks despite jamming,'' in
  \emph{Proceedings of the Twelfth ACM International Symposium on Mobile Ad Hoc
  Networking and Computing}, 2011, pp. 1--10.

\bibitem{golebiewski2009towards}
Z.~Go{\l}{\c e}biewski, M.~Klonowski, M.~Koza, and M.~Kuty{\l}owski, ``Towards
  fair leader election in wireless networks,'' in \emph{International
  Conference on Ad-Hoc Networks and Wireless}.\hskip 1em plus 0.5em minus
  0.4em\relax Springer, 2009, pp. 166--179.

\bibitem{7889052}
D.~{Yu}, L.~{Ning}, Y.~{Zou}, J.~{Yu}, X.~{Cheng}, and F.~C.~M. {Lau},
  ``Distributed spanner construction with physical interference: Constant
  stretch and linear sparseness,'' \emph{IEEE/ACM Transactions on Networking},
  vol.~25, no.~4, pp. 2138--2151, 2017.

\bibitem{garay2015bitcoin}
J.~Garay, A.~Kiayias, and N.~Leonardos, ``The bitcoin backbone protocol:
  Analysis and applications,'' in \emph{Annual International Conference on the
  Theory and Applications of Cryptographic Techniques}.\hskip 1em plus 0.5em
  minus 0.4em\relax Springer, 2015, pp. 281--310.

\end{thebibliography}

\section{Appendix}

\subsection{Chernoff Bound}
\label{sub:chernoff}

\begin{lemma}{(Chernoff Bound).}
\label{lemma:chernoff}
Given a set of independent binary random variables $X_1, X_2,\cdots, X_n$, let $X=\sum_1^n X_i$ and $\mu=\sum_1^n p_i$, where $X_i=1$ with probability $p_i$. If $\mathbb{E}\left[\prod_{i\in S} X_i \leq \prod_{i\in S} q_i \right]$, where $S\subseteq\{0,1,\cdots, n\}$, then it holds for any $\delta>0$ that
\begin{equation*}
\begin{aligned}
P_{r}[X\geq(1+\delta)\mu]\leq e^{-\frac{\delta^{2}\mu}{2(1+\delta/3b)}}.
\end{aligned}
\end{equation*}
If $\mathbb{E}\left[\prod_{i\in S} X_i \geq \prod_{i\in S} q_i \right]$, where $S\subseteq\{0,1,\cdots, n\}$, then for any $\delta\in(0,1]$, we have
\begin{equation*}
\begin{aligned}
P_{r}[X\leq(1-\delta)\mu]\leq e^{-\frac{\delta^{2}\mu}{2}}.
\end{aligned}
\end{equation*}
\end{lemma}

\end{document}